\documentclass[envcountsect,envcountsame]{llncs}
\usepackage{amsmath,amssymb,euscript}
\usepackage{graphicx}
\usepackage{comment}
\usepackage{xspace}
\usepackage{color}
\usepackage[shortlabels]{enumitem}
\newwrite\accuwrite\immediate\openout\accuwrite\jobname.acc
\def\dotheappendixmagic{\immediate\closeout\accuwrite\input\jobname.acc}
\newwrite\accuwrite\immediate\openout\accuwrite\jobname.acc

\CommentEndDef{accumulate}
 \renewenvironment{accumulate}{}{}

\CommentEndDef{onlyaccum}
\newenvironment{onlynoaccum}{}{}
 
 \CommentEndDef{onlynoaccum}

\spnewtheorem{obs}[theorem]{Observation}{\bf}{\it}
\spnewtheorem{conj}[theorem]{Conjecture}{\bf}{\it}

\makeatletter
\def\@fnsymbol#1{\ensuremath{\ifcase#1\or\star\or{\star\!\star}\or
   {\star\!\!\star\!\!\star}\or \dagger\or \ddagger\or
   \mathchar "278\or \mathchar "27B\or \|\or **\or \dagger\dagger
   \or \ddagger\ddagger \else\@ctrerr\fi}}
\makeatother
\def\ca#1{{\cal #1}}

\newcommand{\KK}{\ensuremath{{K}_7 - C_4}\xspace}
\newcommand{\DD}{\ensuremath{\mathcal{D}_3}\xspace}
\newcommand{\EEE}{\ensuremath{\mathcal{E}_5}\xspace}
\newcommand{\FF}{\ensuremath{\mathcal{F}_1}\xspace}
\newcommand{\KKK}{\ensuremath{K_{1,2,2,2}}\xspace}
\newcommand{\EE}{\ensuremath{\mathcal{E}_2}\xspace}
\newcommand{\CC}{\ensuremath{\mathcal{C}_4}\xspace}

\begin{document}
\title{How Not to Characterize Planar-emulable Graphs}
\author{Markus Chimani\inst{1}\thanks{M.~Chimani has been funded by a Carl-Zeiss-Foundation juniorprofessorship.}
	\and Martin Derka\inst{2}\thanks{M.~Derka has been supported by
	 Masaryk University internal grant for students.}
	\and Petr Hlin\v{e}n\'y\inst{2}\thanks{Supported by the Czech
	 science foundation; grants P202/11/0196 and GIG/11/E023.}
	\and Mat\v{e}j Klus\'a\v{c}ek\inst{2}${}^{\mbox{\scriptsize\thefootnote}}$}
\institute{Algorithm Engineering, Friedrich-Schiller-University Jena, Germany\\
\email{markus.chimani@uni-jena.de} \and
Faculty of Informatics, Masaryk University Brno, Czech Republic\\
\email{[hlineny, xderka, xklusac1]@fi.muni.cz}
}
\maketitle

\begin{abstract}
We investigate the question of which graphs have {\em planar emulators}
(a locally-surjective homomorphism from some finite planar graph)---%
a problem raised already in Fellows' thesis (1985) and conceptually related to
the better known planar cover conjecture by Negami (1986).
For over two decades, the planar emulator problem lived poorly in a shadow of
Negami's conjecture---which is still open---as the two were considered
equivalent.
But, in the end of 2008, a surprising construction by Rieck and Yamashita 
falsified the natural ``planar emulator conjecture'', 
and thus opened a whole new research field.
We present further results and constructions which show
how far the planar-emulability concept is from planar-coverability,
and that the traditional idea of likening it to projective embeddability
is actually very out-of-place.
We also present several positive partial characterizations
of planar-emulable graphs.
\end{abstract}

\section{Introduction}
\label{sec:intro}

A graph $G$ has a {\em planar emulator (cover)} $H$ if $H$ is a
finite planar graph and there exists a homomorphism from $H$ onto $G$ that is locally
surjective (bijective, respectively).
In such a case we also say that $G$ is {\em planar-emulable (-coverable)}.
See Def.~\ref{def:emulator} for a precise
definition, and Fig.~\ref{fig:emulcoverexample} for a simple example. 
Informally, every vertex of $G$ is {\em represented} by one or more
vertices in $H$ such that the following holds:
Whenever two nodes $v$ and $u$ are adjacent in~$G$, any node representing $v$ in $H$
has \emph{at least one} (in case of an emulator) or \emph{exactly one} (in case
of a cover) adjacent node in~$H$ that represents~$u$.
Conversely, no node representing $v$ in $H$ has a neighbor representing $u$
if $v,u$ are nonadjacent in~$G$.

Coarsely speaking, the mutually similar concepts of planar covers and planar
emulators both ``preserve'' the local structure of a graph $G$ while ``gaining''
planarity for it. Of course, the central question is which nonplanar graphs do
have planar covers or emulators.

\begin{figure}
$$
\raise 3mm\hbox{\includegraphics[height=18mm]{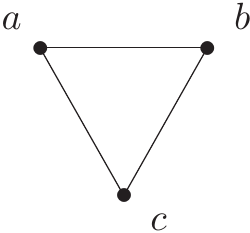}}
\qquad\qquad
\includegraphics[height=23mm]{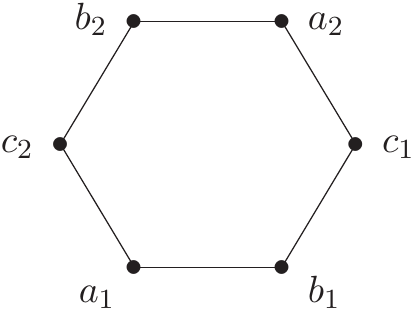}
\qquad\qquad
\lower 2mm\hbox{\includegraphics[height=30mm]{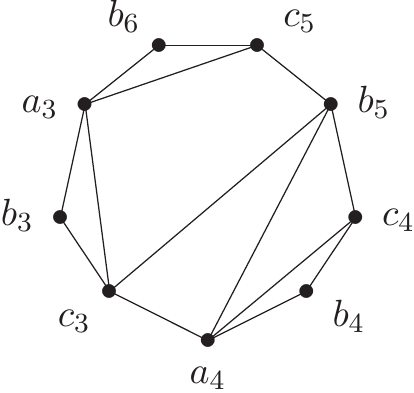}}
\vspace{-1ex}$$
\caption{%
Examples of a planar cover (center) and a planar emulator (right) of the triangle
$G=K_3$ (left).
We simply denote by $a_j$, $j=1,2,\dots$ the vertices representing $a$ of
$G$, and analogically with $b,c$.
}
\label{fig:emulcoverexample}
\end{figure}

The two concepts emerged independently from works of
Fellows~\cite{cit:fellows,cit:femul} (emulator) and 
Negami~\cite{cit:double,cit:conjecture,cit:negsinica} (cover).
On the one hand, the class of planar-coverable graphs is relatively well
understood. At least, we have the following:

\begin{conj}[Negami \cite{cit:conjecture}, 1988]
\label{conj:negami}
A graph has a (finite) planar cover if and only if it embeds in the projective plane.
\end{conj}
Yet, this natural (see below) and firmly believed conjecture is still open today
despite of more than 20 years of intensive research.
See \cite{cit:20years} for a recent survey.

On the other hand, 
it was no less natural to assume \cite{cit:fellows,cit:femul} that
the property of being planar-emulable coincides with 
planar-coverability. By definition, the latter immediately implies the former.
For the other direction, it was highly counterintuitive
to assume that, having more than one neighbors in $H$ representing the same
adjacent vertex of $G$,
could ever help to gain planarity of~$H$\,---such ``additional'' 
edges seem to go against Euler's bound on the number of edges of a planar graph.
Hence, it was widely believed:

\begin{conj}[Fellows \cite{cit:femul}, 1988, falsified 2008]
\label{conj:fellows}
A graph has a (finite) planar emulator if and only if it embeds in the projective plane.
\end{conj}
Perhaps due to similarity to covers, no significant effort to
specifically study planar-emulable graphs occurred during the next 20 years after Fellows' manuscript~\cite{cit:femul}.

Today, however, we know of one important difference between the two
cases: Conjecture~\ref{conj:fellows} is false!
In 2008, Rieck and Yamashita~\cite{cit:rieck} proved
the truly unexpected breakthrough result that there are graphs which have 
planar emulators, but no planar covers and do not embed in the projective plane;
see Theorem~\ref{thm:rieck}.
This finding naturally ignited a new research direction, on which we report
herein. 
We show that the class of planar-emulable graphs is, in fact, 
\emph{much larger} than the class of planar-coverable ones; that the concept of
projective embeddability seems very out-of-place in the context of planar
emulators; and generally, how poorly planar emulators are yet understood.

\medskip
Apart from its pure graph theoretic appeal, research regarding planar emulators and
covers may in fact have algorithmic consequences as well: While Negami's main
interest \cite{cit:double} was of pure graph theoretic nature, Fellows 
\cite[and personal communication]{cit:fellows} considered
computing motivation for emulators.
Additionally, we would like to sketch another potential algorithmic connection; 
there are problems that are NP-hard for general graphs, 
but polynomial-time solvable for planar graphs (e.g., maximum cut), or
where the polynomial complexity drops when considering planar graphs (e.g.,
maximum flow). 
Yet, the precise breaking point is usually not well understood.
Considering such problems for planar-emulable or planar-coverable graphs may
give more insight into the problems' intrinsic complexities. Before this can be
investigated, however, these classes first have to be reasonably well understood
themselves. 
Our paper aims at improving upon this latter aspect of planar emulators.

This paper is organized as follows:
Section~\ref{sec:emulators} discusses all the major prior findings 
w.r.t.\ covers and emulators, including the aforementioned result by Rieck and Yamashita.
Then, Theorem~\ref{thm:main} presents our main new improvement.
Section~\ref{sec:basic} reviews some necessary basic properties and tools,
most of which have been previously sketched in~\cite{cit:femul}.
In Section~\ref{sec:construct} we give previously unknown emulator
constructions, proving Theorem~\ref{thm:main} and also
showing how unrelated emulators are from covers. 
We would particularly like to mention a very small and nicely-structured emulator of
the notoriously difficult graph $K_{1,2,2,2}$ in
\figurename~\ref{fig:K1222emul}.
Finally, in Section~\ref{sec:forbminors} we study how far one can get in the
pursuit to characterize planar-emulable graphs with the structural tools previously
used in \cite{cit:counterex} for covers, and where the current limits are.


\section{On Planar Covers and Emulators}
\label{sec:emulators}

\begin{figure}[tb]\centering
$G=K_5$
\includegraphics[height=33mm]{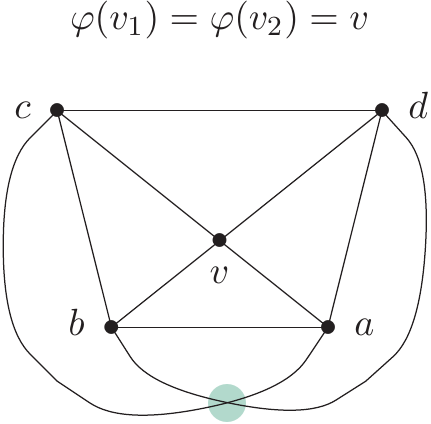}
\qquad\raise 12mm\hbox{\normalsize\boldmath$\longleftarrow\atop\varphi$}\qquad
\includegraphics[height=37mm]{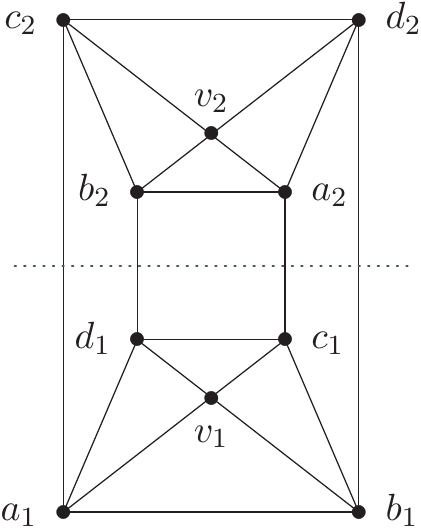}\qquad~
\caption{The graph $G=K_5$ (left) and its two-fold planar cover (right) via
	a homomorphism~$\varphi$.
	The cover is obtained for a ``crosscap-less'' drawing of $G$ and its
	 mirror image.}
\label{fig:coverK5}
\end{figure}

We restate the problem on a more formal level.
All considered graphs are simple, finite, and undirected.
A {\em projective plane} is the simplest nonorientable surface---a plane
with one crosscap (informally, a place in which a bunch of selected edges 
of an embedded  graph may ``cross'' each other).
A graph {\em homomorphism} of $H$ into $G$ is a mapping
$h:V(H)\to V(G)$ such that, for every edge $\{u,v\}\in E(H)$,
we have $\{h(u),h(v)\}\in E(G)$.

\begin{definition}\label{def:emulator}
A graph $G$ has a {\em planar emulator (cover)} $H$ if $H$ is a planar finite graph 
and there exists a graph homomorphism $\varphi:V(H)\to V(G)$ such that,
for every vertex $v\in V(H)$, the neighbors of $v$ in $H$ are mapped by
$\varphi$ surjectively (bijectively) onto the neighbors of $\varphi(v)$ in~$G$.
The homomorphism $\varphi$ is called an {\em emulator (cover) projection}.
\end{definition}
\noindent
One immediately obtains the following two claims:

\begin{lemma}\label{lem:2fold}
a) If $H$ is a planar cover of $G$, then $H$ is also a planar emulator of~$G$.
The converse is not true in general.
\\
b) If $G$ embeds in the projective plane, then $G$ has a two-fold planar cover 
(i.e., $|\varphi^{-1}(u)|=2$ for all $u\in V(G)$); cf.~\cite{cit:double}.
See also \figurename~\ref{fig:coverK5}.
\end{lemma}
These two claims, together with some knowledge about universal coverings in topology, 
make Conjectures~\ref{conj:negami} and~\ref{conj:fellows} sound very plausible.
To precisely describe the motivation for our research direction in planar emulators, 
we briefly comment on the methods that have been used in the investigation
of planar-coverable graphs, too.

\begin{figure}[tb]
\centering
\includegraphics[width=.95\hsize]{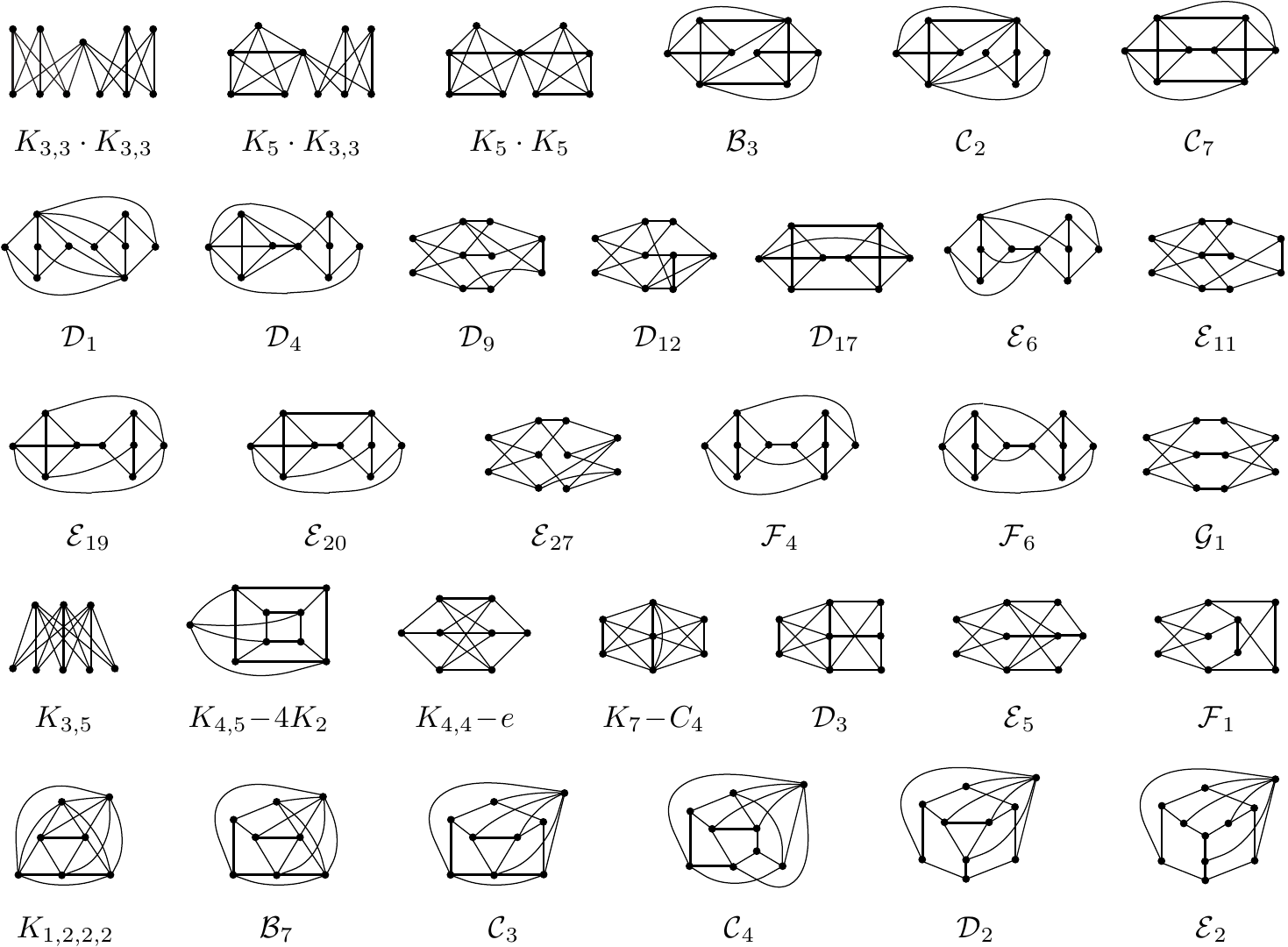}
\caption{The $32$ connected projective forbidden minors. (The three
disconnected ones, $K_5+K_5$, $K_5+K_{3,3}$, $K_{3,3}+K_{3,3}$,
are skipped since they are not important here.)}
\label{fig:32nonproj}
\end{figure}

Firstly, we note that the properties of planar-coverability and planar-emula\-bility
are closed under taking minors (Proposition~\ref{prop:clminor}), 
and all $35$ minor-minimal nonprojective graphs ({\em projective forbidden minors},
\figurename~\ref{fig:32nonproj}) are known~\cite{cit:archd-pp}.
If a connected graph $G$ is projective, then $G$ is planar-coverable (and 
hence also planar-emulable); otherwise, $G$ contains one of the mentioned 
projective forbidden minors. 
Hence to prove Conjecture~\ref{conj:negami}, only a seemingly simple task
remains: we have to show that the known $32$ connected projective forbidden
minors have no planar covers.
The following was established through a series of previous papers:

\begin{theorem}[Archdeacon, Fellows, Hlin\v{e}n\'y, and Negami, 1988--98]
\label{thm:AFHN}
If the (complete four-partite) graph $K_{1,2,2,2}$ has no planar cover,
then Conjecture~\ref{conj:negami} is true.
\end{theorem}

One can naturally think about applying the same arguments to planar emulators,
i.e.\ to Conjecture~\ref{conj:fellows}.
The first partial results of Fellows~\cite{cit:femul}---%
see an overview in Section~\ref{sec:basic}---were, in fact, encouraging.
Yet, all the more sophisticated tools (of structural and discharging
flavor) used to show the non-existence of planar covers in Theorem~\ref{thm:AFHN}
fail on a rather technical level when applied to emulators. 
As these problems seemed to be more of technical than conceptual nature, 
Fellows' conjecture was always believed to be true until the following:

\begin{theorem}[Rieck and Yamashita \cite{cit:rieck}, 2008]
\label{thm:rieck}\hfill
The graphs $K_{1,2,2,2}$ and \mbox{$K_{4,5}-4K_2$}
do have planar emulators (cf.\ Fig.~\ref{fig:K45.3D}).
Consequently, the class of planar-emulable graphs is strictly larger than
the class of planar-coverable graphs, and Conjecture~\ref{conj:fellows} is
false.
\end{theorem}
We remark that this is not merely an existence result, but the actual
(and, surprisingly, not so large) emulators were published together with it.
Both $K_{1,2,2,2}$ and \mbox{$K_{4,5}-4K_2$} are among the
projective forbidden minors, and \mbox{$K_{4,5}-4K_2$} has already 
been proved not to have a planar cover.


\begin{figure}[tb]
\centering
~\raise2mm\hbox{\includegraphics[width=.25\hsize]{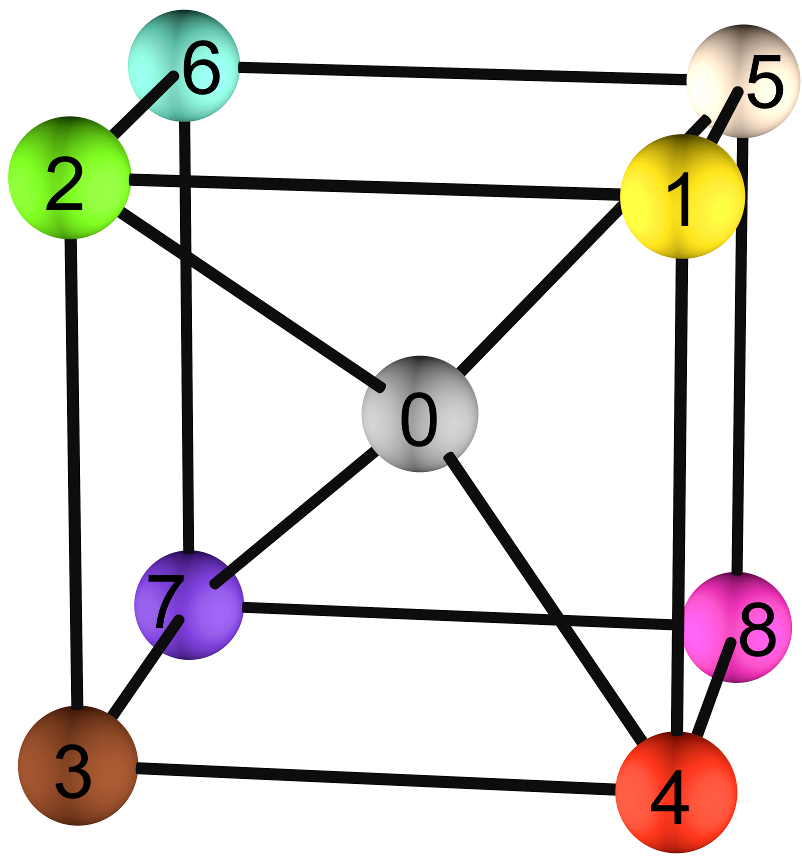}}
	\qquad\raise20mm\hbox{\Large$\leftarrow$}\qquad
	\includegraphics[width=.27\hsize]{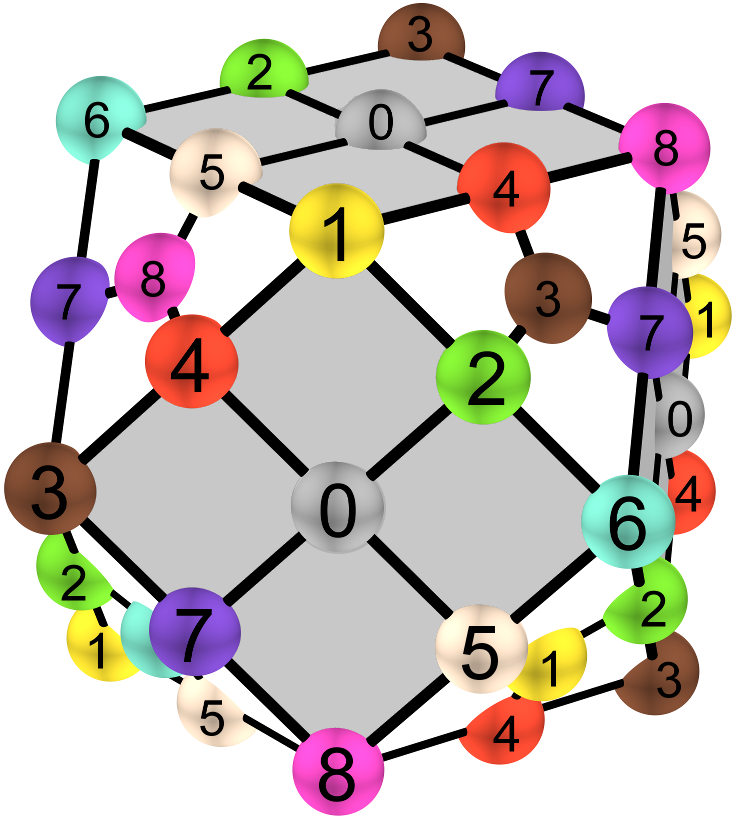}
\caption{A colour-coded 3D-rendering of a planar emulator patched on a polyhedral body (right) for the
	 graph \mbox{$K_{4,5}-4K_2$} (left), taken from
	{\scriptsize\tt http://vivaldi.ics.nara-wu.ac.jp /\~{}yamasita/emulator/}.}
\label{fig:K45.3D}
\end{figure}

One important new message of our paper is that Theorem~\ref{thm:rieck} is
not a rarity---quite the opposite, many other nonprojective graphs have
planar emulators.
In particular we prove that, 
among the projective forbidden minors that have been in
doubt since Fellows'\,\cite{cit:femul}, all except possibly
$K_{4,4}-e$ do have planar emulators:

\begin{theorem}
\label{thm:main}
All of the graphs (\figurename~\ref{fig:32nonproj})
$K_{4,5}-4K_2$, $K_{1,2,2,2}$, $\ca B_7$, $\ca C_3$,
$\ca C_4$, $\ca D_2$, $\ca E_2$, and also $K_7-C_4$, $\ca D_3$,
$\ca E_5$, $\ca F_1$ have planar emulators.
\end{theorem}
Consequently, the class of planar-emulable graphs is much larger than
the class of planar-coverable ones.
We refer to Section~\ref{sec:construct} for details.

\section{Basic Properties of Emulators}
\label{sec:basic}
\begin{onlyaccum}
\subsection{Additions to Section~\ref{sec:basic}}
\end{onlyaccum}

In this section, we review
the basic established properties of planar-emulable graphs.
These are actually all the properties of planar-coverable graphs
which are known to extend to planar emulators
(though, the extensions of some of the proofs are not so straightforward).

The claims presented here, except for Theorem~\ref{thm:K35noemul},
were proved or sketched already in the manuscript \cite{cit:femul} of
Fellows.
However, since \cite{cit:femul} has never been published, we consider it
appropriate to include their full proofs.

We begin with two crucial closure properties.

\begin{proposition}[Fellows \cite{cit:femul}]
\label{prop:clminor}
The property of being planar-emulable is closed under taking \mbox{\em
minors}; i.e., under taking subgraphs and edge contractions.
\end{proposition}
\begin{onlyaccum}
{\def\thetheorem{\ref{prop:clminor}}
\begin{proposition}
The property of being planar-emulable is closed under taking \mbox{minors}; 
i.e., under taking subgraphs and edge contractions.
\end{proposition}}
\end{onlyaccum}

\begin{accumulate}
\begin{proof}
Let $G$ be a planar-emulable graph, and planar $H$ be its emulator via a
projection $\varphi$.
We prove this easy proposition by showing how $H$ is modified to accommodate
for the elementary reduction steps in~$G$;
vertex/edge deletion, and edge contraction.
Say, if a vertex $v\in V(G)$ is deleted, then also all vertices
$\varphi^{-1}(v)$ representing $v$ are deleted from~$H$.

An edge $f=xy\in E(H)$ {\em represents} the edge $e\in E(G)$ if
$e=\{\varphi(x),\varphi(y)\}$.
Whenever an edge $e\in E(G)$ is deleted, so are all the edges 
representing $e$ in~$H$.
Lastly, if an edge $e\in E(G)$ is contracted,
then every component induced by the edges representing $e$ in~$H$ is also
contracted into a single vertex
(note that such components may contain more than one edge representing $e$
in the case of an emulator), and possible parallel edges are simplified.
All these operations preserve planarity of $H$,
and the outcome is an emulator of the graph resulting from~$G$.
\qed\end{proof}
\end{accumulate}

\begin{proposition}[Fellows \cite{cit:femul}]
\label{prop:clYDelta}
The property of being planar-emulable is closed under applying
\mbox{\em $Y\!\Delta$-transformations}; i.e., the operations replacing
(successively) any degree-$3$ vertex with a triangle on its three neighbors.
\end{proposition}
\begin{onlyaccum}
{\def\thetheorem{\ref{prop:clYDelta}}
\begin{proposition}
The property of being planar-emulable is closed under applying
\mbox{\em $Y\!\Delta$-transformations}; i.e., the operations replacing
(successively) any degree-$3$ vertex with a triangle on its three neighbors.
\end{proposition}}
\end{onlyaccum}

\begin{accumulate}
\begin{proof}
Let $G$ be a planar-emulable graph and $v\in V(G)$ a vertex of degree~$3$.
Denote by $G'$ the graph obtained from $G$ by applying the
$Y\!\Delta$-transformation of~$v$.
Suppose a planar graph $H$ that is an emulator of $G$ via a projection $\varphi$.

In the (optimistic) case that all the vertices of $H$ in $\varphi^{-1}(v)$
are also of degree~$3$,
we simply successively apply $Y\!\Delta$-transformations to all the vertices in
$\varphi^{-1}(v)$ (which form an independent set of $H$),
and the resulting graph $H'$ will be again planar and an emulator of $G'$.

It remains to justify our optimistic assumption about degree-$3$ vertices in
$\varphi^{-1}(v)$ of a suitable planar emulator $H$ of~$G$, which
follows from the following claim applied to $X=\{v\}$:
\end{proof}

\begin{lemma}[Fellows \cite{cit:femul}]
\label{lem:cubic_vert}
Let $G$ be a planar-emulable graph
and $X \subseteq V(G)$ an independent set of vertices of degree 3. 
Then there exists a planar emulator $H$ of $G$ with a projection 
$\varphi : V(H) \rightarrow V(G)$ such that every vertex 
$u \in \varphi^{-1}(v)$ over all $v \in X$ is of degree 3. 
\end{lemma}

\begin{proof}
Whenever $F$ is an emulator of our graph $G$ with a projection $\psi: V(F)\rightarrow V(G)$;
let $Dg(F)$ ($\geq3$) shortly denote the maximal $F$-degree of the vertices 
$u \in {\psi}^{-1}(v)$ over all $v \in X$. 
We choose $H$ as a planar emulator of $G$ with projection $\varphi$ such
that the value $Dg(H)$ is minimized.

Assume, for a contradiction, that $Dg(H)>3$, and choose any vertex
$x \in \varphi^{-1}(v)$ where $v \in X$ such that $x$ is of $H$-degree $Dg(H)=d>3$.
Let $a,b,c$ be the three neighbors of $v$ in~$G$.
We denote by $w$ the circular word of length $d$ over the alphabet
$\{a,b,c\}$ formed of the letters
$\varphi(y_1)\varphi(y_2)\dots\varphi(y_d)$, where $y_1,\dots,y_d$ are the
neighbors of $x$ in~$H$ in this cyclic order.
Then, one of the following three cases, up to symmetry, occurs in $w$:

\begin{figure}[t]
  \centering
  \def \svgwidth{\columnwidth}
  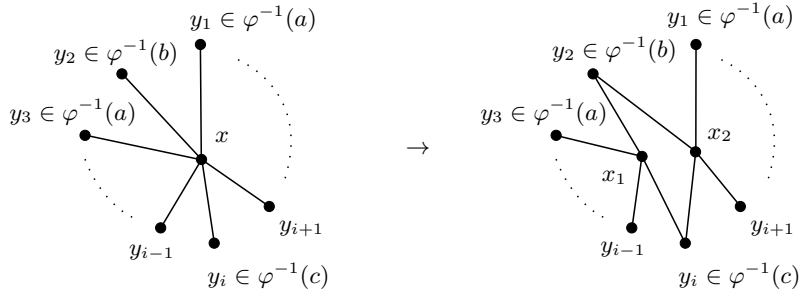
  \caption{Splitting vertex $x$ with a cubic image into vertices of lower degree.}
  \label{fig:cubic_vert_case2}
\end{figure}
\begin{figure}[tb]
  \centering
  \def \svgwidth{\columnwidth}
  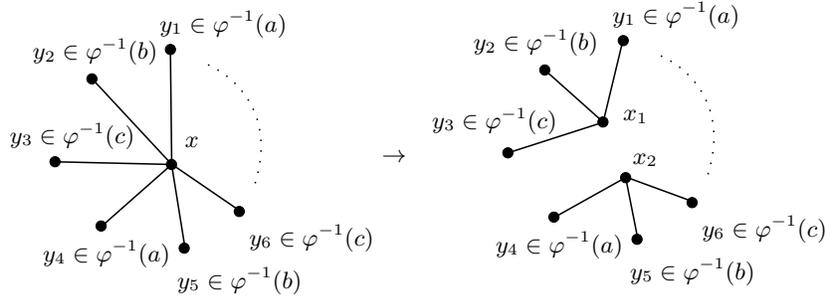
  \caption{Illustration of the last case of the proof of Lemma~\ref{lem:cubic_vert}.}
  \label{fig:cubic_vert_case3}
\end{figure}

\begin{itemize}
    \item 
	$w$ contains a subword $aa$:
	By merging the corresponding two vertices of $H$ representing $a$ into one, 
	the degree of $x$ drops to $d-1$.
    \item 
	$w$ contains a subword $aba$:
	Without loss of generality, it is $\varphi(y_1)=\varphi(y_3)=a$,
 $\varphi(y_2)=b$, and $\varphi(y_i)=c$ for some $4\leq i\leq d$
 (to be a valid emulator of $G$).
 We modify $H$ by splitting vertex $x$ into $x_1,x_2$ with
 $\varphi(x_1) = \varphi(x_2) = \varphi(x)$,
 so that $x_1$ is adjacent to $y_2,y_3,\dots,y_i$ and $x_2$ to
 $y_i,\dots,y_d,y_1,y_2$; see \figurename~\ref{fig:cubic_vert_case2}.
 Clearly, the degrees of $x_1,x_2$ are now smaller than~$d$.
    \item $w = (abc)^+$:
 Then $H$ may be modified as shown in Fig~\ref{fig:cubic_vert_case3},
 and the degrees of the newly created vertices drop down to~$3$.
\end{itemize}

In each of the cases it is easy to see that the obtained graph $H'$ is still 
a valid planar emulator of $G$,
and that only degrees of some neighbors of $x$ in $H$ 
could have gone up from $H$ to $H'$.
Hence, as $X$ is an independent set, we can repeat the above construction for
all the vertices $x\in\varphi^{-1}(X)$ (which form an independent set in
$H$, too) of degree $d$, and in finitely many steps obtain a contradiction
to minimality of $Dg(H)$.
\qed\end{proof}
\end{accumulate}

Next, we identify some easy forbidden minors for planar-emulable graphs
among the known list of projective forbidden minors (cf.\
Lemma~\ref{lem:2fold}\,b).
Again, these extend folklore knowledge about planar-coverable graphs.

\begin{accumulate}
We say that a graph $G$ contains {\em two disjoint k-graphs}
if there exist two vertex-disjoint subgraphs $J_1,J_2\subseteq G$
such that, for $i=1,2$, the graph $J_i$ is isomorphic to a subdivision of
$K_4$ or $K_{2,3}$, the subgraph $G-V(J_i)$ is connected and adjacent to
$J_i$, and contracting in $G$ all the vertices of $V(G)\setminus V(J_i)$ into one
results in a nonplanar graph (i.e.\ containing a $K_5$- or
$K_{3,3}$-subdivision).
We remark that such $G$ is always nonprojective \cite{cit:proj103}.
See an example in \figurename~\ref{fig:2kgraphs}.

\begin{figure}[tb]
\centering
\includegraphics[width=.4\hsize]{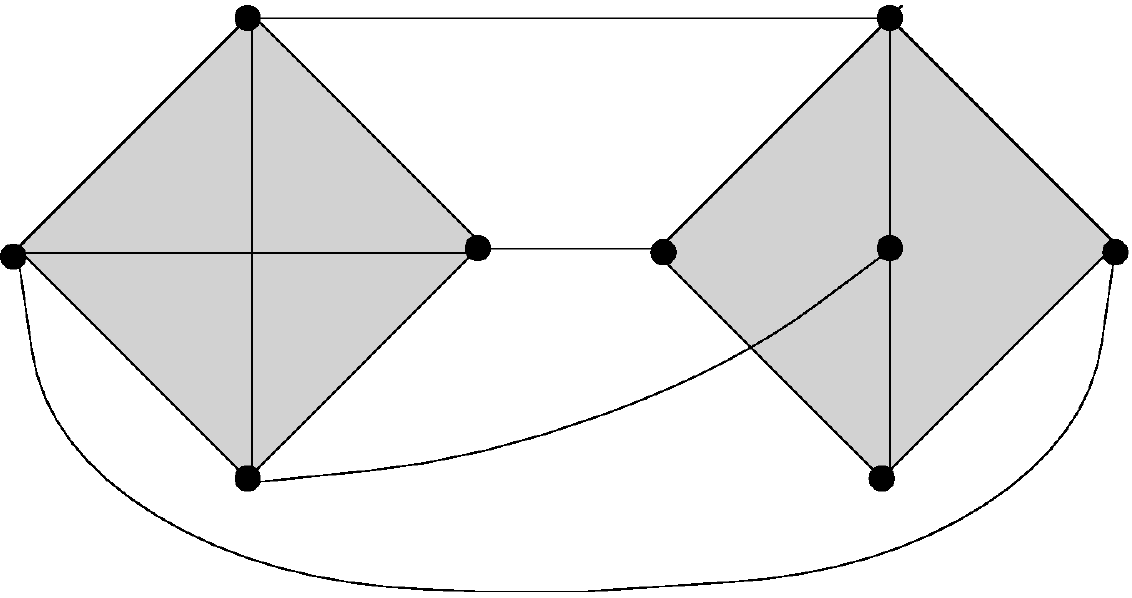}
\caption{An example of a graph having two disjoint k-graphs (shaded in gray).}
\label{fig:2kgraphs}
\end{figure}

{%
\begin{onlynoaccum}
\def\thetheorem{\ref{thm:2kgraphs}}
\end{onlynoaccum}
\begin{theorem}[Fellows \cite{cit:femul}]
\label{thm:2kgraphs}
A planar-emulable graph $G$ cannot contain two disjoint \mbox{k-graphs}.
Consequently, each of the $19$ graphs---projec\-tive forbidden minors---%
in the first three rows of \figurename~\ref{fig:32nonproj} 
has no planar emulator.
\end{theorem}
}

\begin{proof}
Suppose, for a contradiction, that $G$ contains two disjoint k-graphs
$J_1,J_2\subseteq G$,
and that there exists a planar emulator $H$ with a projection $\varphi:V(H)\to V(G)$.
Let $H_i$, $i=1,2$, denote the subgraph of $H$ induced by the edges
representing $E(J_i)$ in the projection $\varphi$.
(An edge $f=xy\in E(H)$ {\em represents} $e$ if $e=\{\varphi(x),\varphi(y)\}$.)
Then $H_1$ and $H_2$ are vertex-disjoint, and up to symmetry between
$H_1,H_2$, there exists a component $A_1\subseteq H_1$ such that all other
components of $H_1,H_2$ lie in the outer face of $A_1$ in the plane drawing
of~$H$.

Since $G-V(J_1)$ is connected and adjacent to
$J_1$, it follows that all the vertices of 
$V(H)\setminus V(A_1)$ lie in the outer face of $A_1$.
So, by contracting $V(H)\setminus V(A_1)$ into one vertex $x$ we obtain a planar
graph $H_0$ which is an emulator of the nonplanar graph $G_0$ resulting from
$G$ by contracting all $V(G)\setminus V(J_1)$ into one vertex $w$.
Let $\varphi_0:V(H_0)\to V(G_0)$ be the derived emulator projection.
Then $\varphi_0^{-1}(w)=\{x\}$, which is a contradiction to further
Lemma~\ref{lem:enonplanar}.
\qed\end{proof}

\begin{lemma}[Fellows \cite{cit:femul}]
\label{lem:enonplanar}
In every planar emulator $H$ of a nonplanar connected graph $G$
with the projection $\varphi:V(H)\to V(G)$, the following holds:
$|\varphi^{-1}(v)|\geq2$ for each $v\in V(G)$.
\end{lemma}

\begin{proof}
Suppose, for a contradiction, that $\varphi^{-1}(w)=\{x\}$ for some 
$w\in V(G)$ and $x\in V(H)$.
Firstly, we prove the claim for $G=K_5$:
Then $H-x$ is an emulator of $K_4=K_5-w$, and $H-x$ is outerplanar,
i.e.\ all its vertices are incident with one face since they are all adjacent to
the same vertex $x$ in $H$.
However, all degrees in $H-x$ are at least $3$ while an outerplanar simple graph
must contain a vertex of degree $\leq2$, a contradiction.

Secondly, we consider $G=K_{3,3}$ which is a bit more complicated case.
Then $H-x$ is an emulator of $G-w=K_{2,3}$.
Obviously, $H-x$ may be assumed connected.
Let $B$ be a leaf block of $H-x$, i.e.\ a maximal $2$-connected subgraph of
$H-x$ such that $B$ shares only (at most) one vertex with the rest of $H-x$.
Let $\{a,b,c\}\subseteq V(K_{2,3})$ denote the unique independent set
of size three, and $\{s,t\}\subseteq V(K_{2,3})$ be the other two vertices.
Then every vertex of $H-x$ representing $s$ or $t$ is of degree $\geq3$,
and there exists such $y\in V(B)$ having all neighbors $z_1,z_2,z_3$ in $B$.
Since $z_1,z_2,z_3$ are mapped to $a,b,c$, they must all be adjacent to $x$
in plane $H$, which contradicts $2$-connectivity of~$B$.

Third, we consider any other nonplanar graph $G$, i.e.\ containing a
minor isomorphic to $K_5$ or $K_{3,3}$.
Notice in the proof of Proposition~\ref{prop:clminor} that even a minor $G'$ of 
$G$ will have an emulator $H'$ (a minor of $H$) with projection $\varphi'$ such that
$|\varphi'^{-1}(w')|=1$ for $w'$ corresponding to original~$w$.
Hence we are finished by one of the previous two cases.
\qed\end{proof}

\end{accumulate}

\begin{onlynoaccum}
\begin{theorem}[Fellows \cite{cit:femul}]
\label{thm:2kgraphs}
A planar-emulable graph cannot contain ``two disjoint k-graphs''
(see the appendix).
Consequently, each of the $19$ graphs---projec\-tive forbidden minors---%
in the first three rows of \figurename~\ref{fig:32nonproj} 
has no planar emulator.
\end{theorem}
\end{onlynoaccum}

\begin{accumulate}
Finally, we include the following sporadic result
which seems to be just a {\em very fortunate} extension of the cover case,
heavily benefiting from Lemma~\ref{lem:cubic_vert}.
\end{accumulate}

\begin{theorem}[Fellows\,/\,Huneke~\cite{cit:hunekecov}]
\label{thm:K35noemul}
The graph $K_{3,5}$ has no planar emulator.
\end{theorem}
\begin{onlyaccum}
{\def\thetheorem{\ref{thm:K35noemul}}
\begin{theorem}[Fellows \cite{cit:femul}]
The graph $K_{3,5}$ has no planar emulator.
\end{theorem}}
\end{onlyaccum}

\begin{accumulate}
\begin{proof}
Suppose, for a contradiction, that the graph $K_{3,5}$ has a planar
emulator~$H$ with a projection $\varphi:V(H)\to V(K_{3,5})$, 
and denote by $X\subseteq V(K_{3,5})$ the subset of degree-$3$ vertices in $K_{3,5}$.
By Lemma~\ref{lem:cubic_vert}, we may assume that all the vertices in $H$
representing some vertex of $X$ are of degree~$3$ as well.
Furthermore, since a homomorphic image of an odd cycle contains an odd cycle
but $K_{3,5}$ is bipartite,
the emulator $H$ is also bipartite.
Hence the overall setting is (almost) as in the cover case and we may apply
arguments analogical to \cite{cit:hunekecov,cit:20years}.

We use the so called discharging method.
We assign {\em charge} of $3(4-deg(x))$ to every vertex $x$ of degree $deg(x)$,
and of $3(4-len(f))$ to every face $f$ of length $len(f)$ in~$H$.
Note that $len(f)\geq4$ is always even in $H$.
By Euler's formula, the total charge of $H$ is positive $12\cdot2>0$.
The aim of the discharging method is to redistribute this charge across $H$
in a way that the resulting amount is nonpositive,
which would give a contradiction to supposed planarity of~$H$.

Subsequently, every degree-$3$ vertex of $H$ (i.e., every vertex representing one of $X$)
sends its charge equally $1$ to each neighbor.
Then any vertex $y\in V(H)$ of degree $d\geq6$ or more ends up with total charge
of at most $3(4-d)+d=12-2d\leq0$.
On the other hand, every degree-$5$ vertex $z\in V(H)$ now has charge of $-3+5=2$.
That charge is subsequently sent from $z$ to any incident face
of length $\ell\geq6$ in $H$,
which then ends up with charge of at most $3(4-\ell)+\ell=12-2\ell\leq0$.
This gives the required contradiction provided we can show that not all
faces incident with $z$ are of length~$4$.

\begin{figure}[tb]
\centering
\includegraphics[width=.75\hsize]{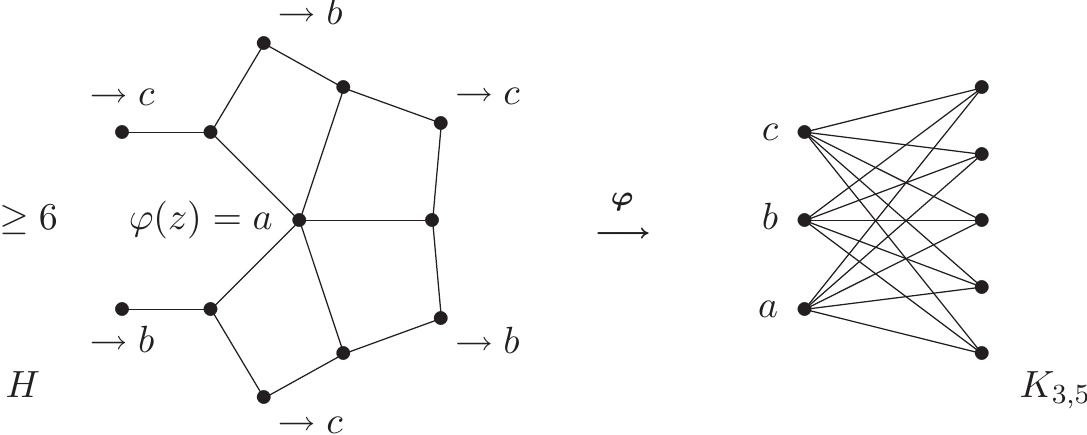}
\caption{An illustration of the proof of Theorem~\ref{thm:K35noemul};
	not all faces incident with the central vertex $z$ can be of length~$4$.}
\label{fig:K35cov}
\end{figure}

Now assume the latter,
and denote by $\{a,b,c\}=V(K_{3,5})\setminus X$ such that $\varphi(z)=a$.
See \figurename~\ref{fig:K35cov}.
The neighbors of $z$ in $H$ are all of degree $3$, and each one of them needs one
additional neighbor representing $b$ and one representing $c$.
So the vertices in the second neighborhood of $z$ in $H$ 
alternatingly represent $b,c,b,c,\dots$,
and we thus cannot have exactly five of them incident to faces of
length~$4$ around $z$.
This contradiction proves that some face incident to $z$ is of length
$\geq6$, as needed to finish the proof.
\qed\end{proof}

\end{accumulate}

Lastly, we remark that also the graphs $K_7$ and $K_{4,4}$ cannot have
planar emulators by Euler's formula, but these may not be minor-minimal such ones.
In particular, there is some (yet unknown) subgraph of the complete graph $K_7$ which is
a minor-minimal non-planar-emulable graph, as discussed in
Section~\ref{sec:conclus}.


\section{Constructing New Planar Emulators}
\label{sec:construct}
\begin{onlyaccum}
\subsection{Additions to Section~\ref{sec:construct}}
\end{onlyaccum}

The central part of this paper deals with new constructions of planar
emulators which consequently give the proof of Theorem~\ref{thm:main}.
\begin{onlynoaccum}
In this section we sketch the interesting
(and in some sense central) emulators for
the graphs \EE and \KK (\figurename~\ref{fig:32nonproj}),
while a more detailed description together with emulators for the rest
of the graphs discussed in Theorem~\ref{thm:main} can be found in the Appendix. 
\end{onlynoaccum}
We remark that, to our best knowledge, no planar emulators of nonprojective 
graphs other than those mentioned in Theorem~\ref{thm:rieck} have been 
studied or published prior to our paper.
Moreover, using our systematic techniques we have succeeded in
finding a much smaller emulator for \KKK than the one presented by Rieck and
Yamashita~in~\cite{cit:rieck}.


\subsubsection*{Planar emulator for \EE.}
In order to obtain an easily understandable description of an emulator for \EE,
we note the following:
A graph isomorphic to \EE (in \figurename~\ref{fig:32nonproj}) can be constructed
from the complete graph $K_4$ on $V(K_4)=\{1,2,3,4\}$
by subdividing each edge once, calling the new vertices \emph{bi-vertices},
and finally introducing a new vertex $0$ adjacent to all the bi-vertices.

\begin{figure}[tb]\centering
\includegraphics[height=3.6cm]{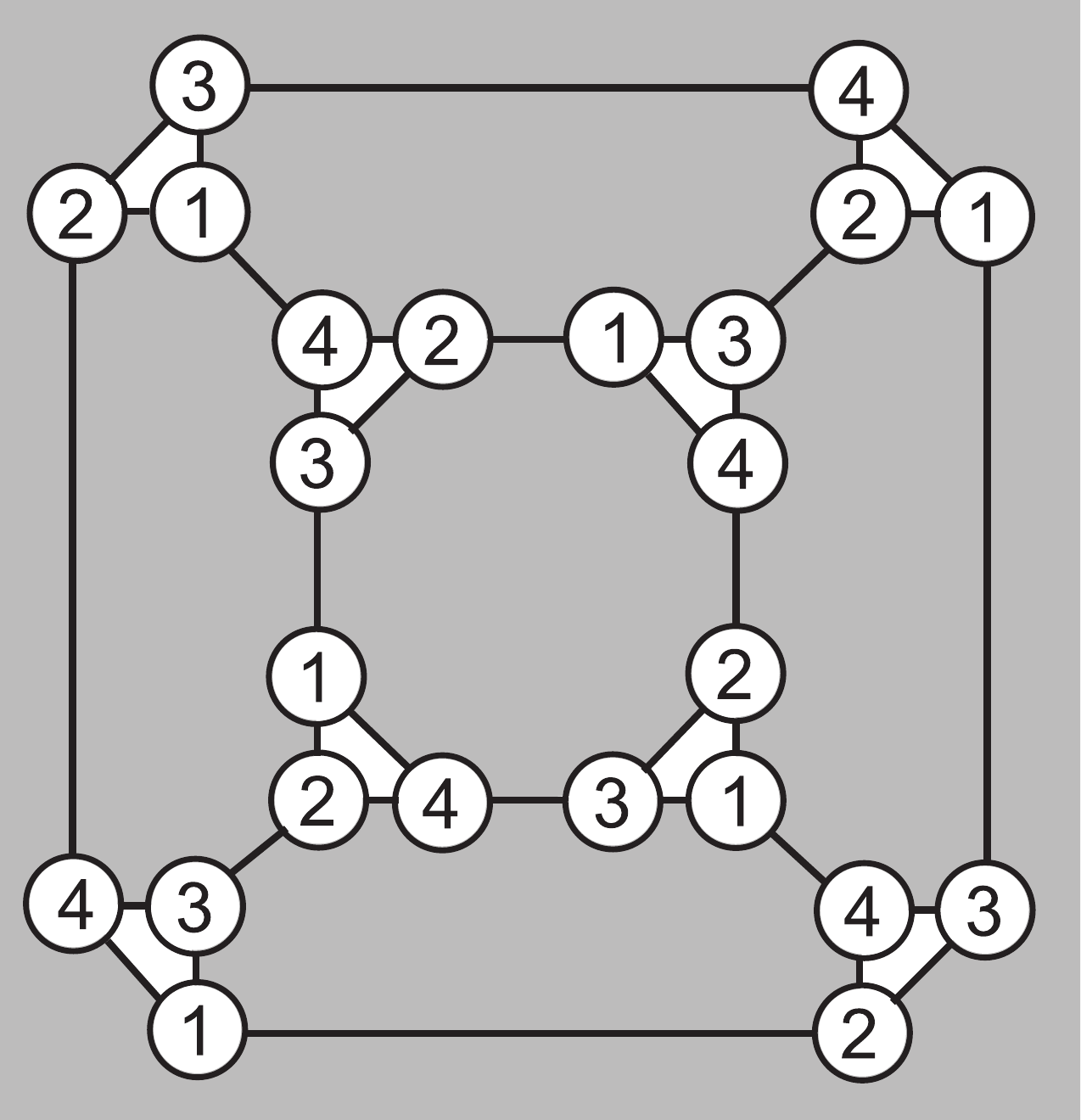} \qquad\qquad
\includegraphics[height=3.6cm]{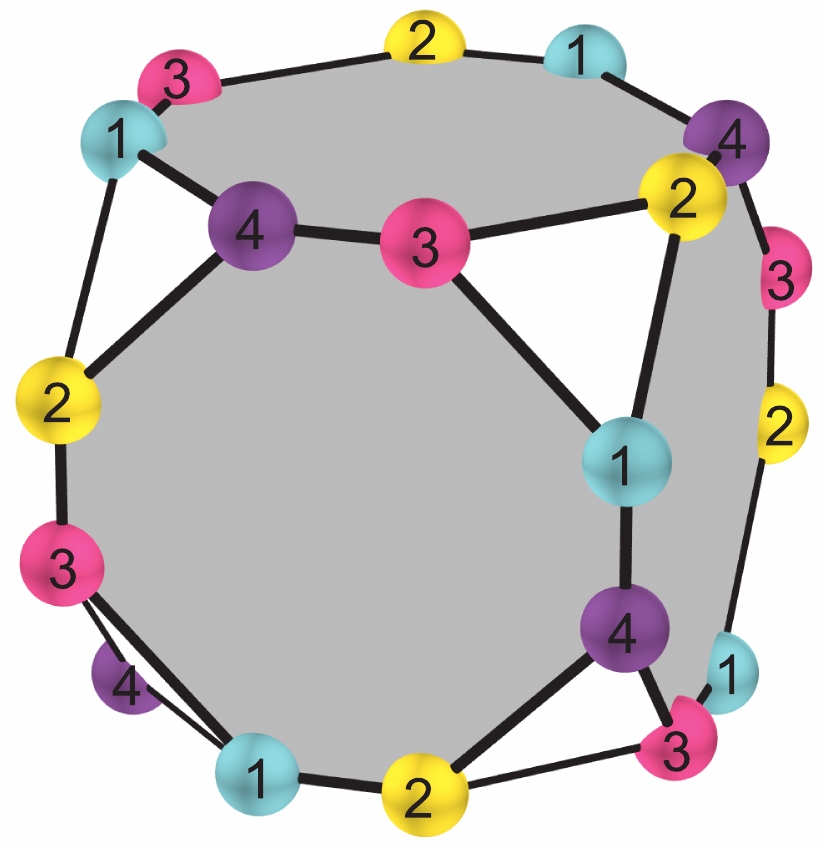}
\caption{A planar emulator (actually, a cover) for the complete graph 
	$K_4$ with the rich faces depicted in gray colour. 
	The same figure in a ``polyhedral'' manner on the right.}
\label{fig:e2helper}
\end{figure}

A similar sketch can be applied to a construction of a planar emulator for~\EE:
If one can find a planar emulator for $K_4$ with the additional property
that each edge is incident to at least one \emph{rich} face---i.e., a face
bordered by representatives of all edges of $K_4$,
then a planar emulator for \EE can be easily derived from this.
More precisely, if $H_0$ is such a special emulator of $K_4$,
see an example in \figurename~\ref{fig:e2helper},
then the following construction is applied.
Each edge of $H_0$ is subdivided with a new vertex representing the
corresponding bi-vertex of \EE,
and a new vertex representing the vertex $0$ of \EE is added to every rich
face of $H_0$ such that it is adjacent to all the 
subdividing vertices within this face.
The resulting plane graph $H$ clearly is an emulator for \EE
(and this construction is reversible).

\begin{figure}[tb]\centering
\includegraphics[height=7cm]{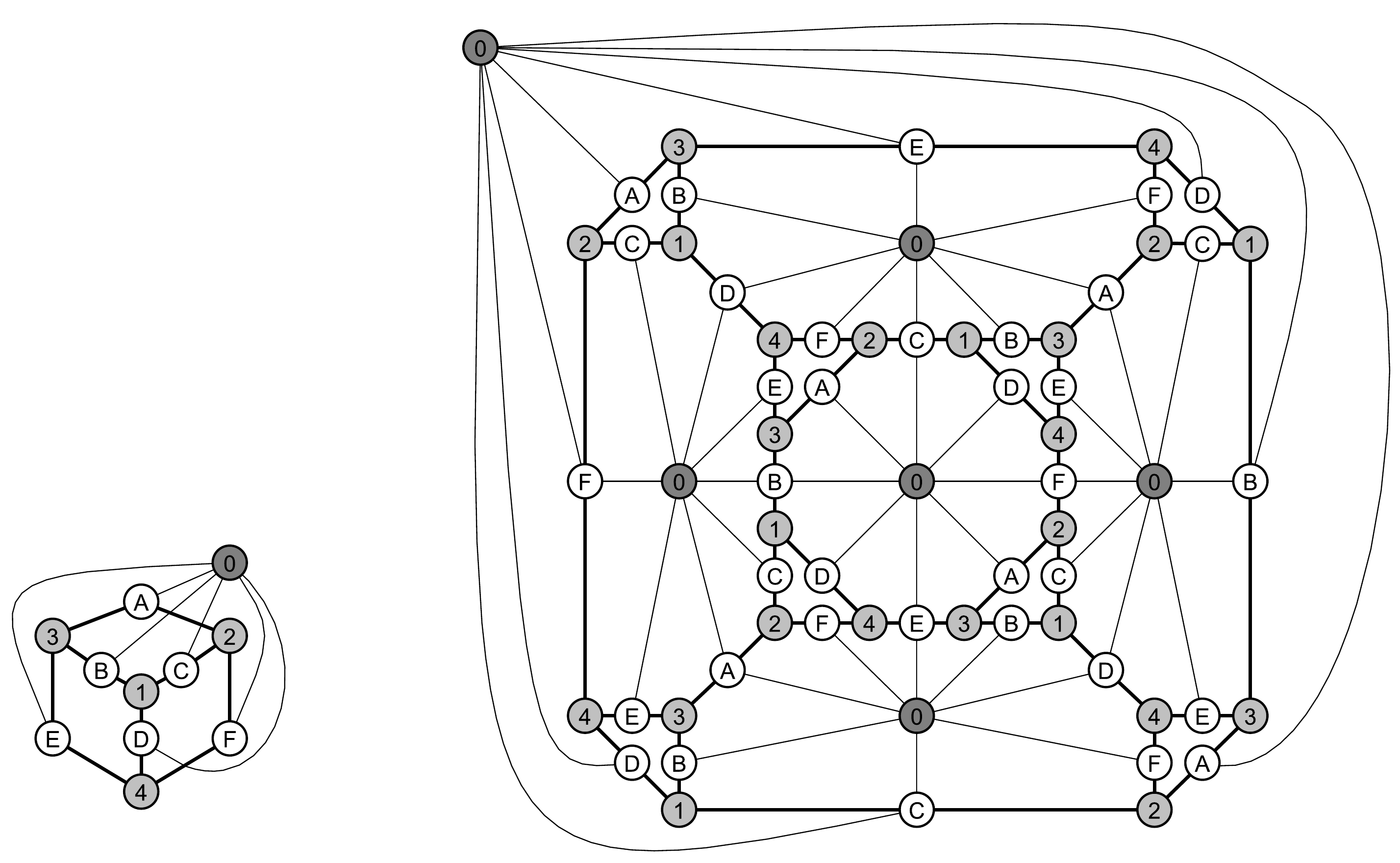}
\caption{A planar emulator for \EE.
	The bi-vertices of the construction are in white and labeled with
	letters, while the numbered core vertices 
	 (cf.~\figurename~\ref{fig:e2helper}) are in gray.}
\label{fig:e2_final}
\end{figure}

\begin{figure}[htbp]\centering
\includegraphics[height=7cm]{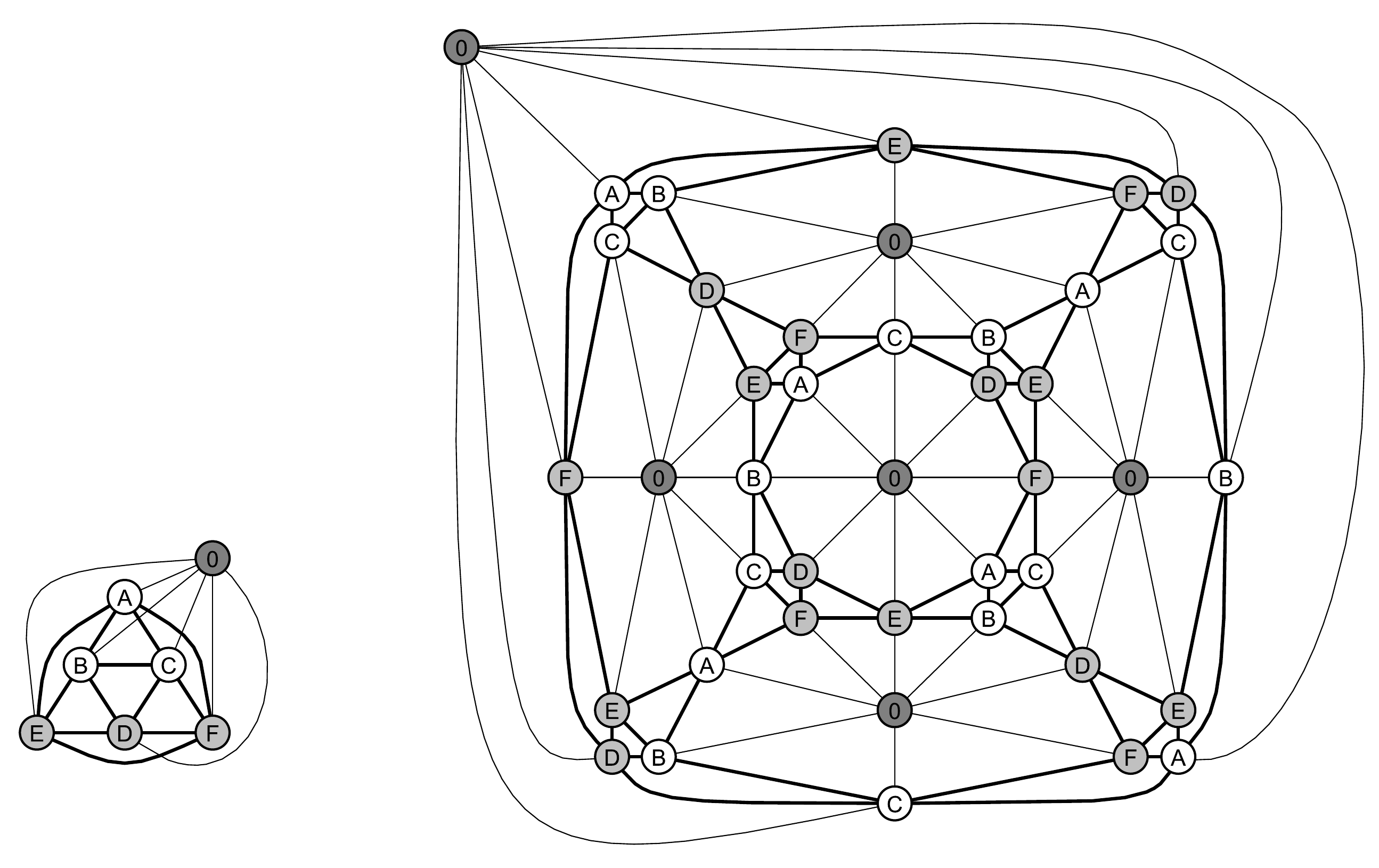}
\caption{A planar emulator for \KKK; 
	obtained by taking {\ensuremath{Y\!\Delta}}-transformations on
	the core vertices labeled $1,2,3,4$ of the \EE emulator from
	\figurename~\ref{fig:e2_final}.}
\label{fig:K1222emul}\label{fig:k1222_final}
\end{figure}


Perhaps the simplest possible such an emulator for $K_4$ with rich faces is
depicted in \figurename~\ref{fig:e2helper} (left).
This leads to the nicely structured planar emulator for the graph \EE in
\figurename~\ref{fig:e2_final}.
It is also worth to note that the same core ideas which helped us to find
this emulator for \EE, were actually used in \cite{cit:c4e2} to prove the
{\em nonexistence} of a planar cover for \EE.
This indicates how different the coverability and emulability concepts are
from each other, too.


\begin{onlyaccum}
\subsubsection*{Planar emulator for \EE; additional notes.}
\end{onlyaccum}
\begin{accumulate}
\medskip
There is another interesting point to mention about our emulator for \EE
---the plane graph can be quite beautifully pictured as a polyhedron
(compare to \figurename~\ref{fig:e2helper} right).
Consider a cube; it has 8 corners, 12 ridges (we avoid the term \emph{edge}
here), and 6 facets. Interpreting the corners, ridges, and facets of any
convex polyhedron as the vertices, edges, and faces gives a planar graph;
geometrically, we obtain a plane drawing of the cube graph 
by choosing a perspective projection from a point close to one of the cube facets.

Then we may truncate (``cut'') each of the eight corners of the cube
(geometrically, to obtain a \emph{truncated hexahedron}, an {\bf Archimedean solid}),
and represent each of the eight $6$-cycle (but triangle-shaped) faces of the
emulator from \figurename~\ref{fig:e2_final} at each of the truncated
corners.
We place, among those $6$-faces, pairs of the same type at the
opposite corners of the cube. 
Then we add the respective missing edges along the cube ridges, 
and finally we place the remaining vertices representing $0$ into the
six octagonal facets of the body, which correspond to the rich gray faces
from \figurename~\ref{fig:e2helper}.

\end{accumulate}


\subsubsection*{More emulators derived from the \EE case.}
\begin{onlyaccum}
\subsubsection*{More emulators derived from the \EE case.}
\end{onlyaccum}

By Proposition~\ref{prop:clYDelta}, the property of having a planar
emulator is closed under taking {\ensuremath{Y\!\Delta}}-transformations. 
Moreover, the proof is constructive, and we may use it to mechanically
produce new emulators from existing ones (this principle goes even slightly beyond
straightforward ${Y\!\Delta}$-transformations, see Section~\ref{sec:forbminors}).
Therefore we can easily obtain an alternative emulator for \KKK
(cf.~Theorem~\ref{thm:rieck}) which is significantly smaller and simpler
than the original one in~\cite{cit:rieck}.
The emulator is presented in \figurename~\ref{fig:k1222_final}.

\begin{figure}[htb]
\centering
\includegraphics[height=7cm]{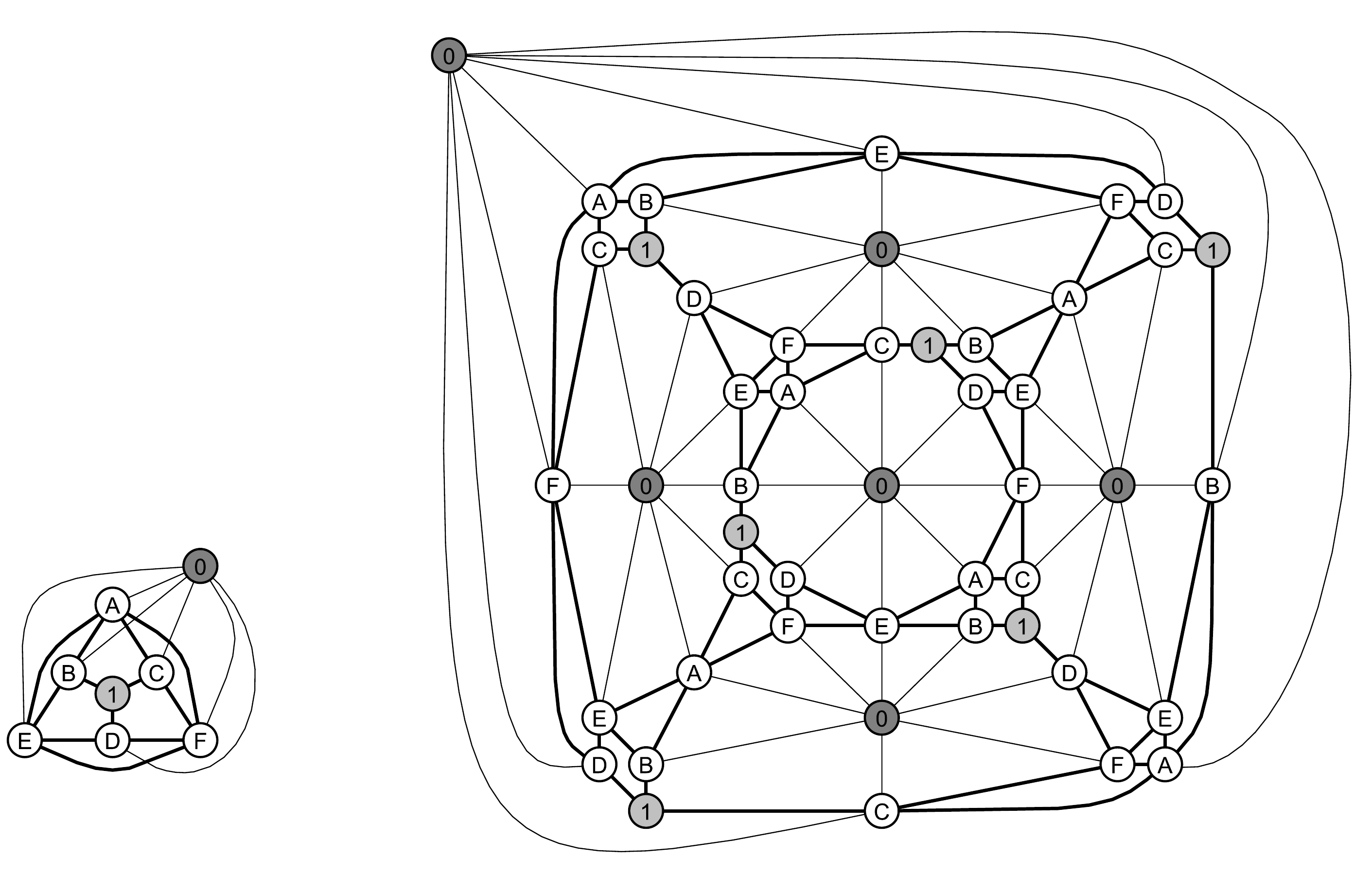}
\caption{Emulator for \ensuremath{\mathcal{B}_7} }
\label{fig:B7_final}
\end{figure}
\begin{figure}[htb]
\centering
\includegraphics[height=7cm]{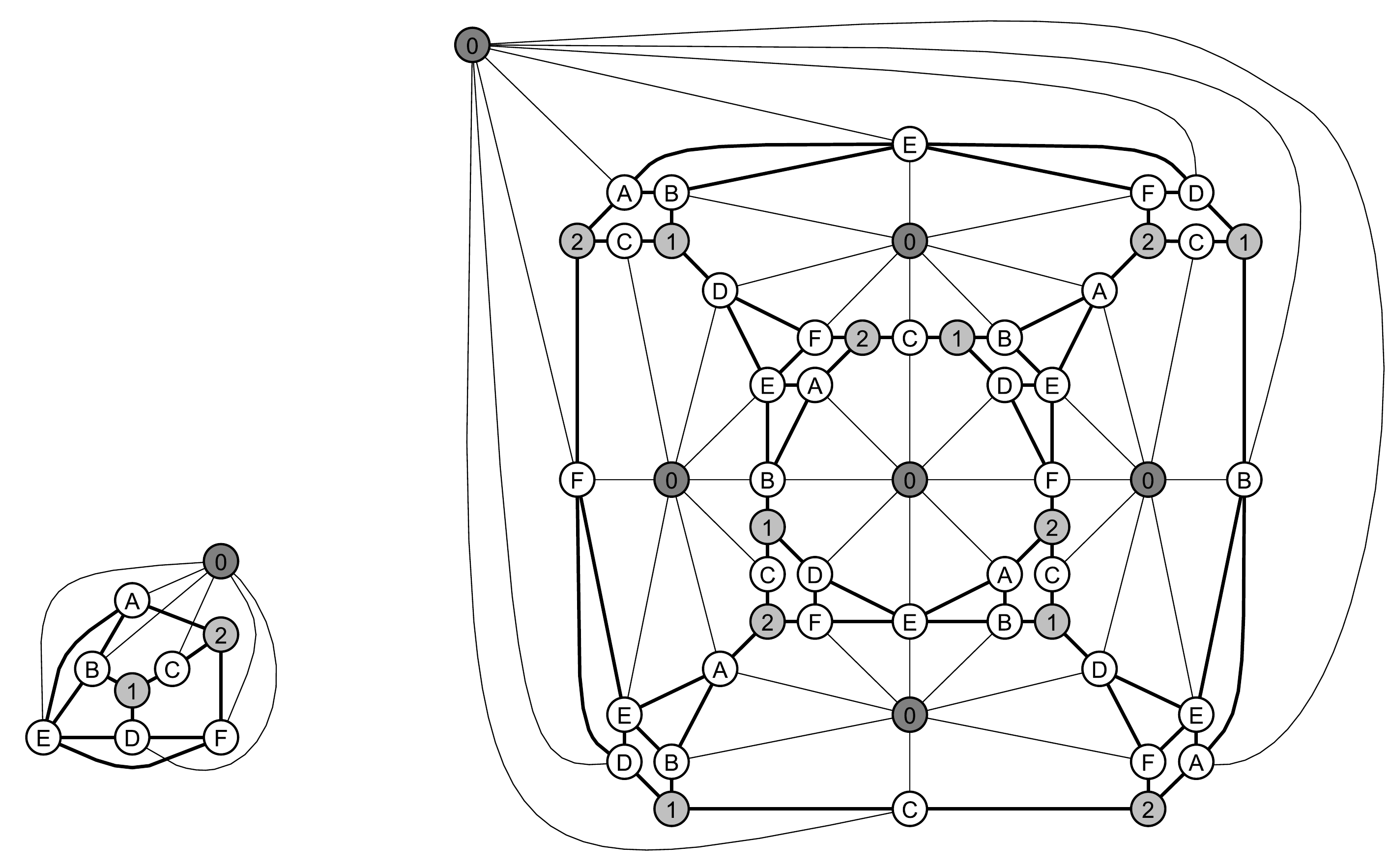}
\caption{Emulator for \ensuremath{\mathcal{C}_3}}
\label{fig:C3_final}
\end{figure}
\begin{figure}[htb]
\centering
\includegraphics[height=7cm]{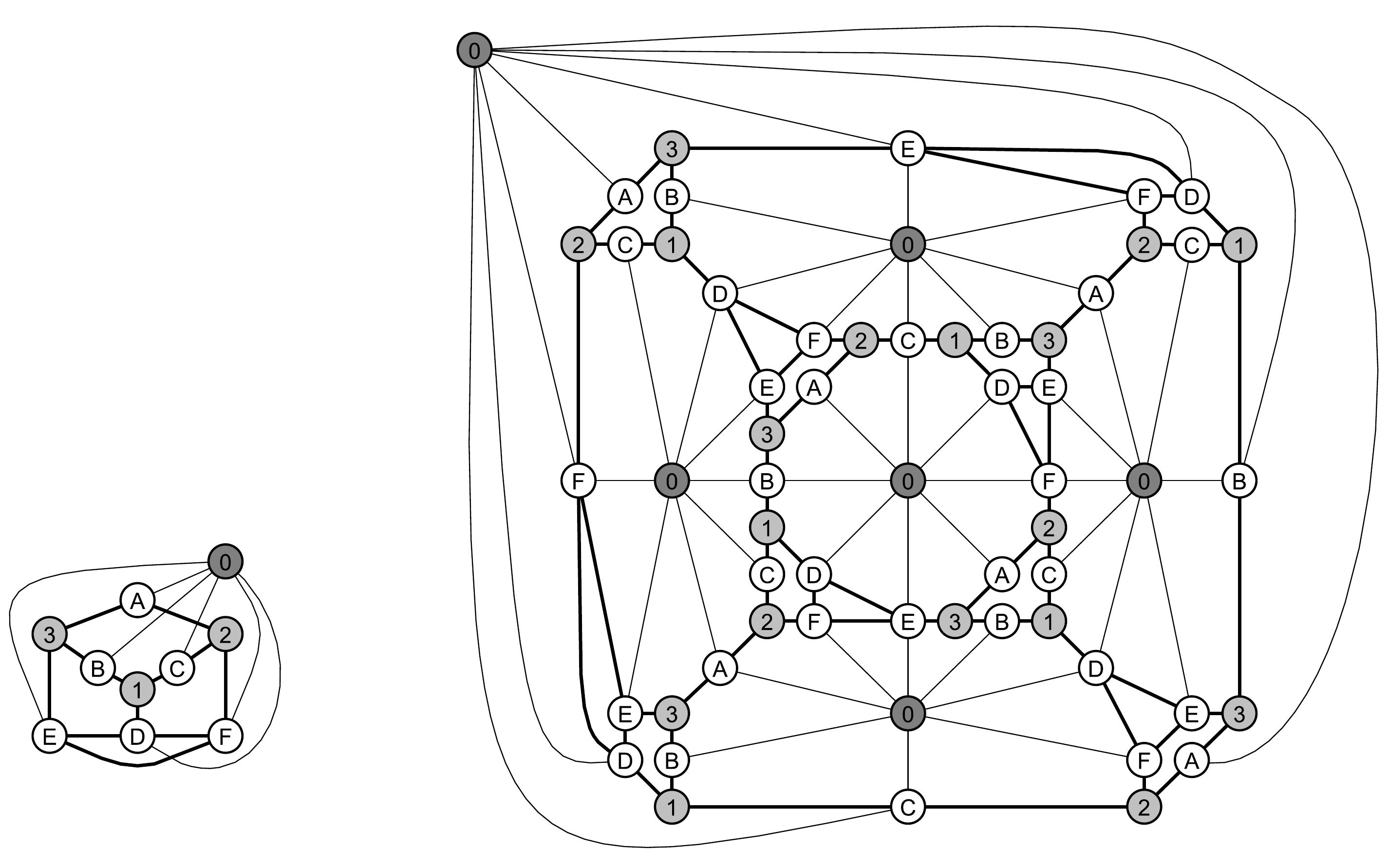}
\caption{Emulator for \ensuremath{\mathcal{D}_2}}
\label{fig:D2_final}
\end{figure}

Furthermore, in the same mechanical way, we can obtain planar emulators for
other members of the ``\KKK-family''; namely for $\ca B_7$, $\ca C_3$, $\ca D_2$
in \figurename~\ref{fig:32nonproj}.
\begin{accumulate}
Several more interesting planar emulators can be straightforwardly 
obtained from that of \EE by means of ${Y\!\Delta}$-transformations.
See these emulators in
Figures~\ref{fig:B7_final},\ref{fig:C3_final},\ref{fig:D2_final}.

\end{accumulate}
\begin{onlynoaccum}
On the other hand, finding a planar emulator
for the last member, $\ca C_4$, seems to be a more complicated case---the
smallest one currently has $338$ vertices and we
postpone its description to the appendix.
\end{onlynoaccum}

\begin{accumulate}
\subsubsection{Planar emulator for \CC.}
Consider the graph \CC drawn and labeled as in Figure~\ref{fig:c4},
and observe that it is constructed of the cube graph with
all nodes except for two (say $0$ and $7$) in the opposing corners of the cube
adjacent to an additional vertex $x$.

\begin{figure}[htdp]
\centering\bigskip
\includegraphics[height=3cm]{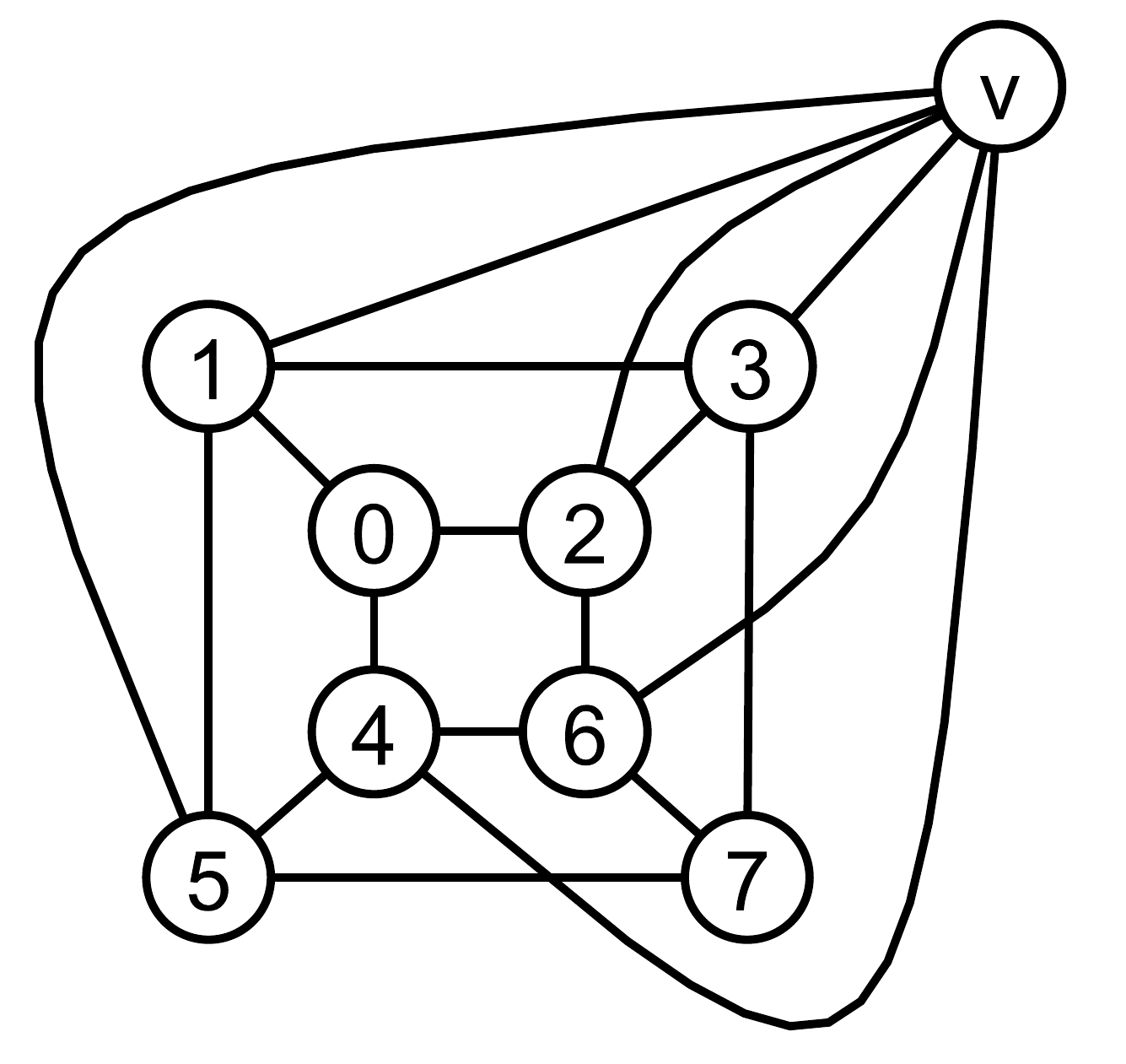}
\caption{The graph \CC.\label{fig:c4}}
\end{figure}

Figure~\ref{fig:c4gadget} shows the gadget we will be using: We can think of it as the
trace that arises when rolling the underlying cube over its ridges.
We start (north-west of the gadget) with the cube lying on the facet 
$\{0,1,2,3\}$, and roll it along its $\{1,3\}$ ridge, such that it lies down
with the facet $\{1,5,7,3\}$. 
Overall, we roll the cube seven times around this axis,
i.e., each possible side is downwards exactly twice; we end up at the 
north-east of the gadget. There, we change the roll-axis,
and roll over the ridge $\{6,2\}$. Again we roll seven times and arrive at the 
south of the gadget. There, we change the roll axis again, and, after seven rolls, 
arrive back at our start position. 

\begin{figure}[tbp]
\centering
\includegraphics[height=4cm]{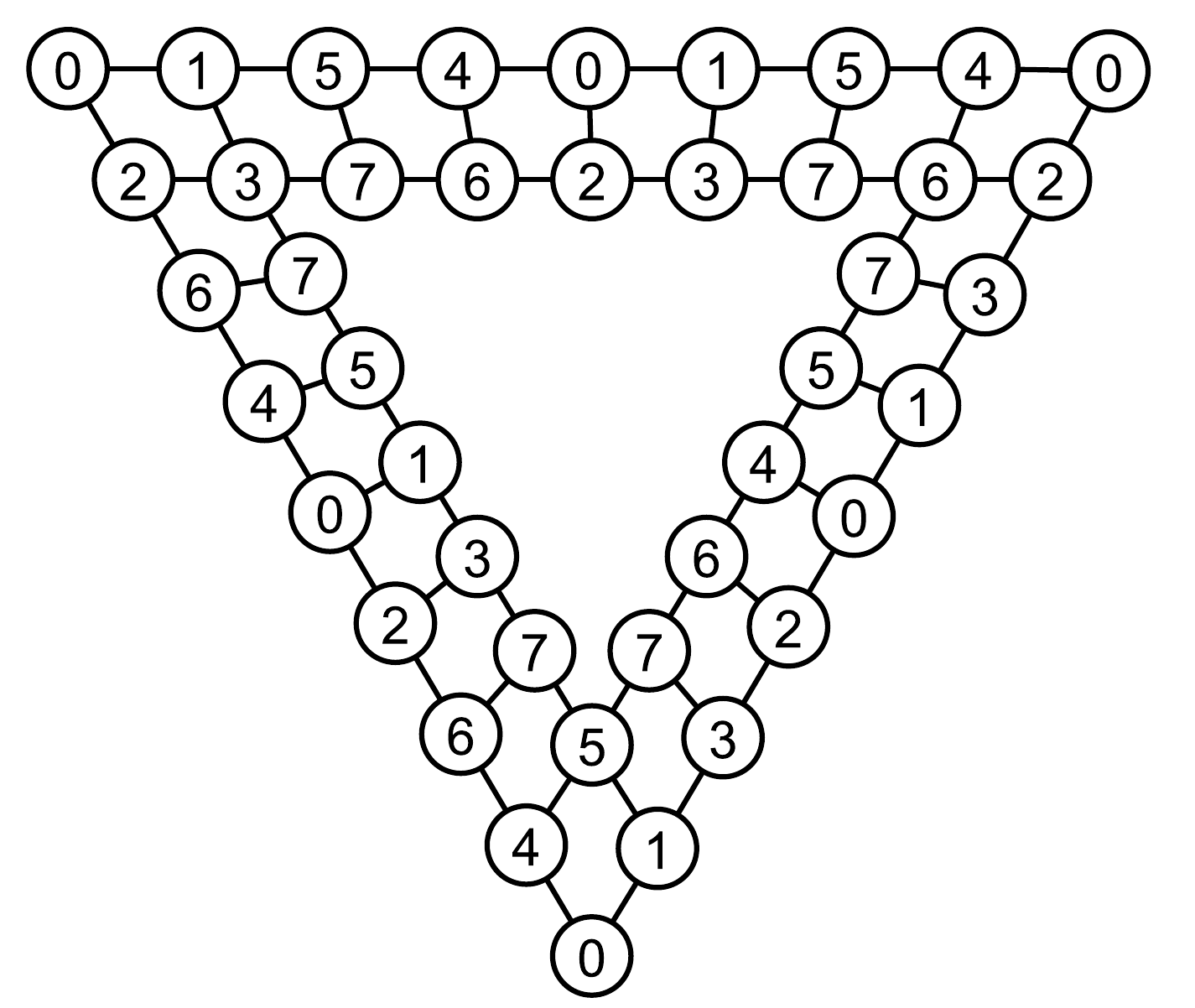}
\caption{Gadget used to build an emulator for \CC.\label{fig:c4gadget}}
\end{figure}

\begin{figure}[htb]
\centering
\includegraphics[height=9cm]{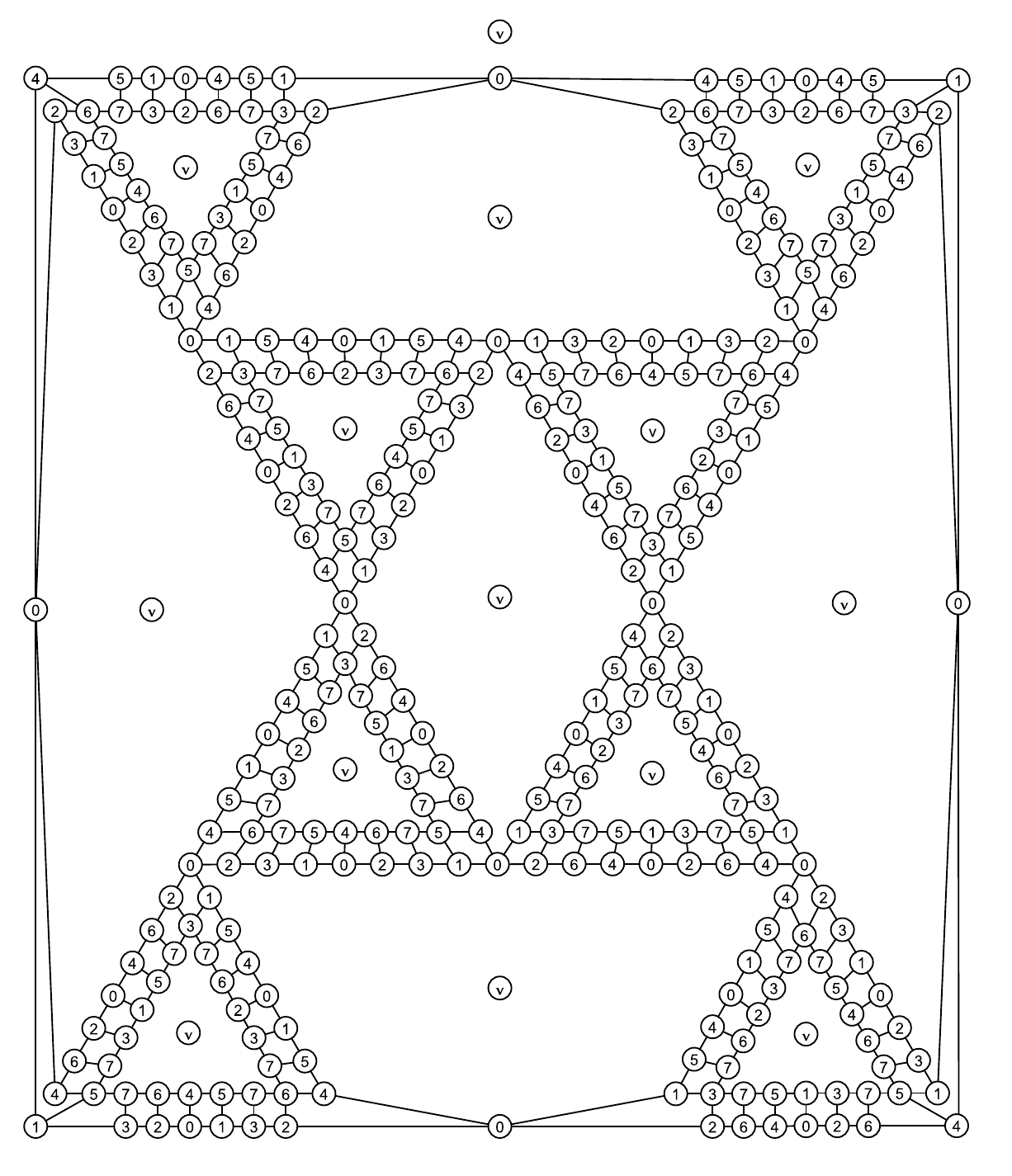}
\caption{The full planar emulator for \CC.\label{fig:c4emul}}
\end{figure}

The arising, triangular-shaped gadget
allows an intuitive notion of \emph{outside} (nodes on the outer face) and 
\emph{inside} (all other nodes; they lie on the largest inner face). It
has several important properties: 
\begin{itemize}
\item The node $0$ only appears on the outside of the gadget. Each such node 
misses exactly one of the neighbors required for~\CC.
\item All nodes $1$--$6$
have degree three, and are adjacent to all the necessary neighbors (w.r.t.\ \CC),
except for $x$.
\item The node $7$ only appears on the inside of the gadget, and is adjacent to
all its necessary neighbors (w.r.t.\ \CC). 
\item Connecting all nodes on the inside (outside, respectively) of the gadget 
with an additional vertex $x'$ representing $x$ suffices for $x'$ to satisfy
its emulator property for \CC.
\item On the outside of the gadget, the node $5$ ($6$, $3$) appears
only on the north (south-west, south-east, respectively) side.
\end{itemize}
We complete the gadget by inserting a node representing $x$ into the inside of 
the gadget.

Now, to obtain an emulator for \CC, we construct a graph embedded on a 
\emph{cuboctahedron} (the Archimedean solid with 8 triangular and 6 square 
facets): clearly, we can draw its wire-frame structure planarly.
Note that each of the polyhedron's ridges is neighbored by one triangular and 
one square facet.
We label all corners of this polyhedron with $0$, and insert a (properly rotated, 
see below) copy of our gadget into each of the polyhedron's triangular facets. 
We can uniquely label the ridges of the polyhedron with $5$, $6$, or $3$, depending
on which of these nodes appear on the gadget's side along that ridge. All nodes $0$ have their required neighbors, and we can place a node representing $x$ into each square face and
connect it with the nodes along its ridges. This establishes an emulator for $\CC$,
cf.\ Figure~\ref{fig:c4emul}.

\end{accumulate}

\subsubsection*{Planar emulator for \KK.}

Already the survey \cite{cit:20years}---when commenting on the surprising
Rieck--Yamashita construction---stressed the importance of deciding
whether the graph \KK is planar-emulable.
Its importance is tied with the structural search for all potential
nonprojective planar-emulable graphs;
see \cite{cit:counterex,cit:martinbc} and Section~\ref{sec:forbminors} for
a detailed explanation.
Briefly saying, \KK (and its ``family'' of \DD, \EEE, \FF;
\figurename~\ref{fig:32nonproj}) are the only projective forbidden minors 
which have planar emulators and are not ``internally $4$-connected''.
In fact, for several reasons we believed that \KK cannot have a planar emulator,
and so it came as another surprise when we have just recently discovered
one.

\begin{figure}[tb] \centering
\vspace{-3mm}
\includegraphics[height=2.7cm]{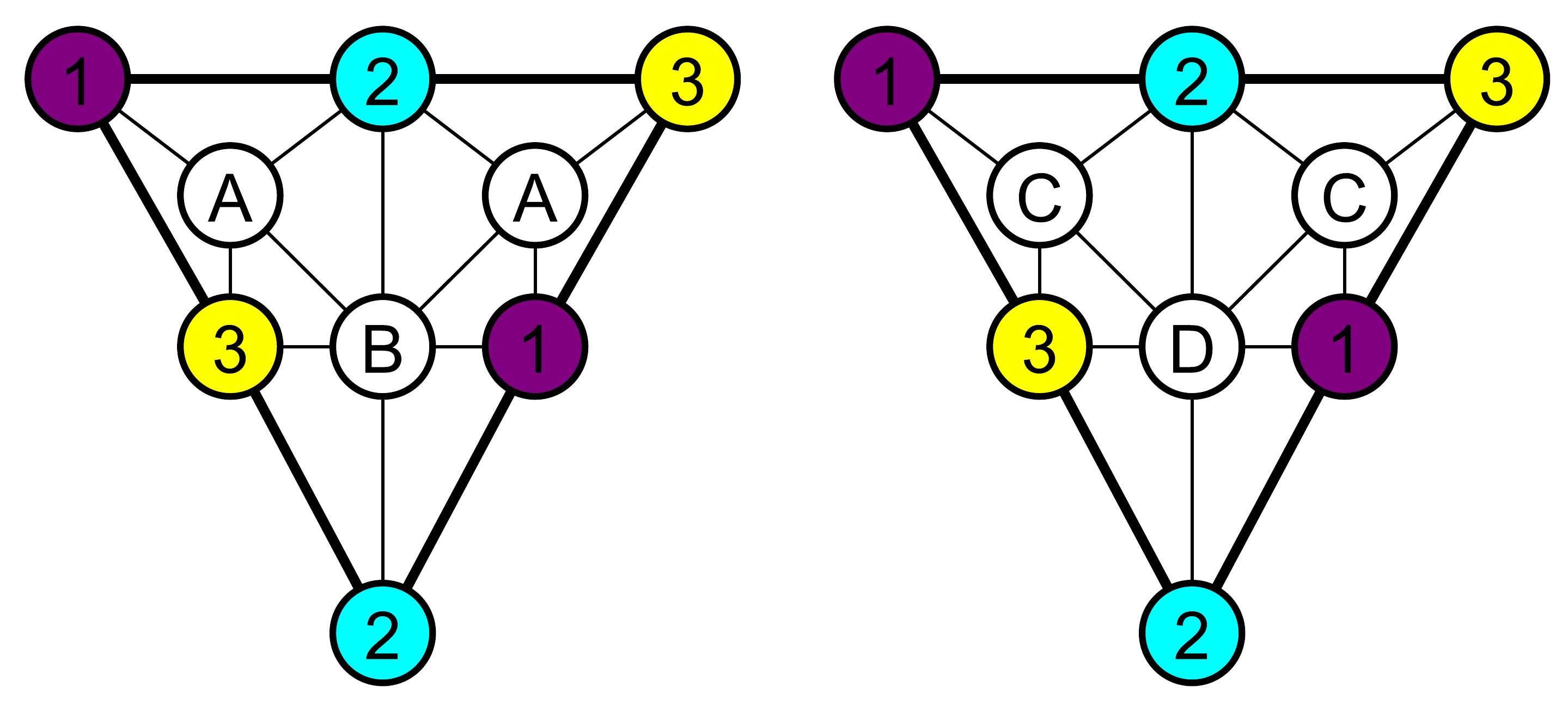}
 \qquad\qquad
\includegraphics[height=3cm]{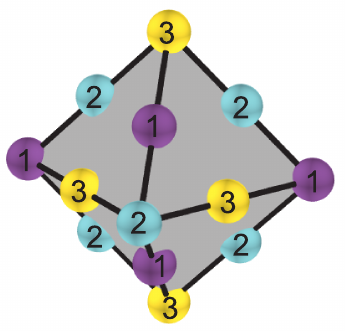}
\caption{Basic building blocks for our \KK planar emulator:
	On the left, only vertex 2 misses an A-neighbor and 1,3 miss a  B-neighbor.
	Analogically on the right. The right-most picture shows the skeleton of the emulator in 
	a ``polyhedral'' manner.}
\label{fig:K7-C4_cells}
\end{figure}

\begin{figure}[tb]\centering
\includegraphics[height=7cm]{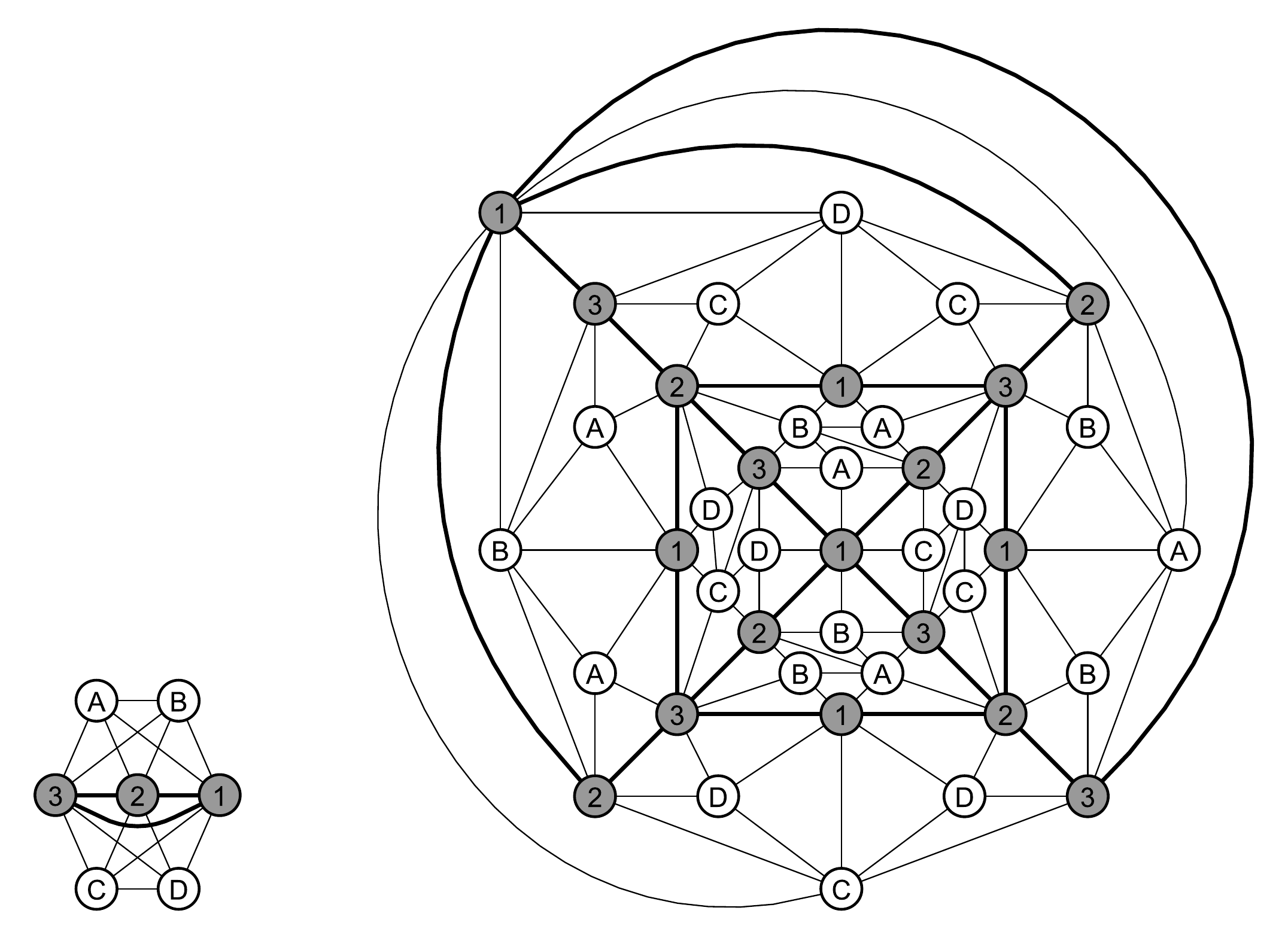}
\caption{A planar emulator for \KK, constructed from the blocks in
	\figurename~\ref{fig:K7-C4_cells}.
	The skeleton representing the central vertices is drawn in bold.}
\label{fig:K7-C4_final}
\end{figure}

In order to describe our planar emulator construction for  \KK,
it is useful to divide the vertex set of \KK into three groups:
the triple of \emph{central vertices} (named
$1,2,3$ in \figurename~\ref{fig:K7-C4_final} left) adjacent to all 
other vertices, and the two vertex pairs (named $A,B$ and $C,D$)
each of which has connections only to its mate and to the central triple.
This view allows us to identify a \emph{skeleton} of the potential emulator
as the subgraph induced on the vertices representing the central
triple~$1,2,3$ and place the remaining vertices representing $A,B$ and
$C,D$ into the skeleton faces, provided certain additional requirements are
met.

This simple idea leads to the introduction of basic building blocks (see
Figure~\ref{fig:K7-C4_cells}), each of which 
``almost'' emulates the subgraph induced on 1,2,3,A,B and 1,2,3,C,D, respectively.
The crucial property of the blocks is that the vertices labeled A,B or C,D have
all the required neighbors in place.
Finally, four copies of each of the blocks can be arranged in the shape of
an \emph{octahedron} such that all missing requirements in the blocks are
satisfied.
The resulting planar emulator is in Figure~\ref{fig:K7-C4_final}.

\begin{onlynoaccum}
Similar, though much more involved, procedures lead to constructions
of planar emulators for the graphs \DD, \EEE, \FF
(which are $Y\!\Delta$-transformable to \KK).
Those emulators have 126, 138, and 142 vertices, respectively,
and we refer interested readers to the appendix.
\end{onlynoaccum}

\begin{accumulate}
\subsubsection*{Presenting the planar emulator for \KK; additional notes.}

Emulator for \KK has a similar property as the one for \EE. We can
embed it into a polyhedron---an octahedron in this case.  We may then
 take 8 cells from Figure~\ref{fig:K7-C4_cells} and call
them \textit{AB/CD cells}.  Three out of six outer vertices of each cell
have both inner vertices as their neighbors (they are \textit{AB or
CD-satisfied}) and the remaining three have one (they are \textit{AB or
CD-half-satisfied}).  We take four cells of each kind and join the outer
vertices such that every \textit{AB/CD cell} will represent one facet of an
octahedron and no two cells of the same kind will be adjacent.  Notice that
the \textit{central vertices} in the middle of each ridge of the octahedron
are incident to two facets and \textit{central vertices} on the corners of
the octahedron are incident to four facets.  We now rotate the \textit{AB/CD
cells} such that vertices on corners are twice \textit{AB-half-satisfied}
and twice \textit{CD-half-satisfied}, every time by a different vertex, and
vertices on the ridges are \textit{AB-satisfied} by one of two incident
facets and \textit{CD-satisfied} by the other.  Such a construction is an
emulator for \KK and can surely be drawn planarly (see
Figure~\ref{fig:K7-C4_final}).

\subsubsection{Planar emulator for \DD.}

Once we can emulate \KK, the natural question to ask is if this construction
can be extended to \DD, a graph which is created by applying a single
{\ensuremath{\Delta Y}} transformation on \KK (replacing one triangular face
of \KK by a vertex of degree 3) - see Figure~\ref{fig:D3}.  Clearly, the
same construction does not work, because of a special property of
vertex 1, which will be discussed later.  Again, we call again three vertices
labeled 1,2,3 the \textit{central vertices} and two other components
\textit{ABC and DE}.  We will consider an extra edge between
vertices 1 and 3, because it does not influence the property of having an
emulator in this case (see section 5).

\begin{figure}[htbp]\centering
\includegraphics[height=3cm]{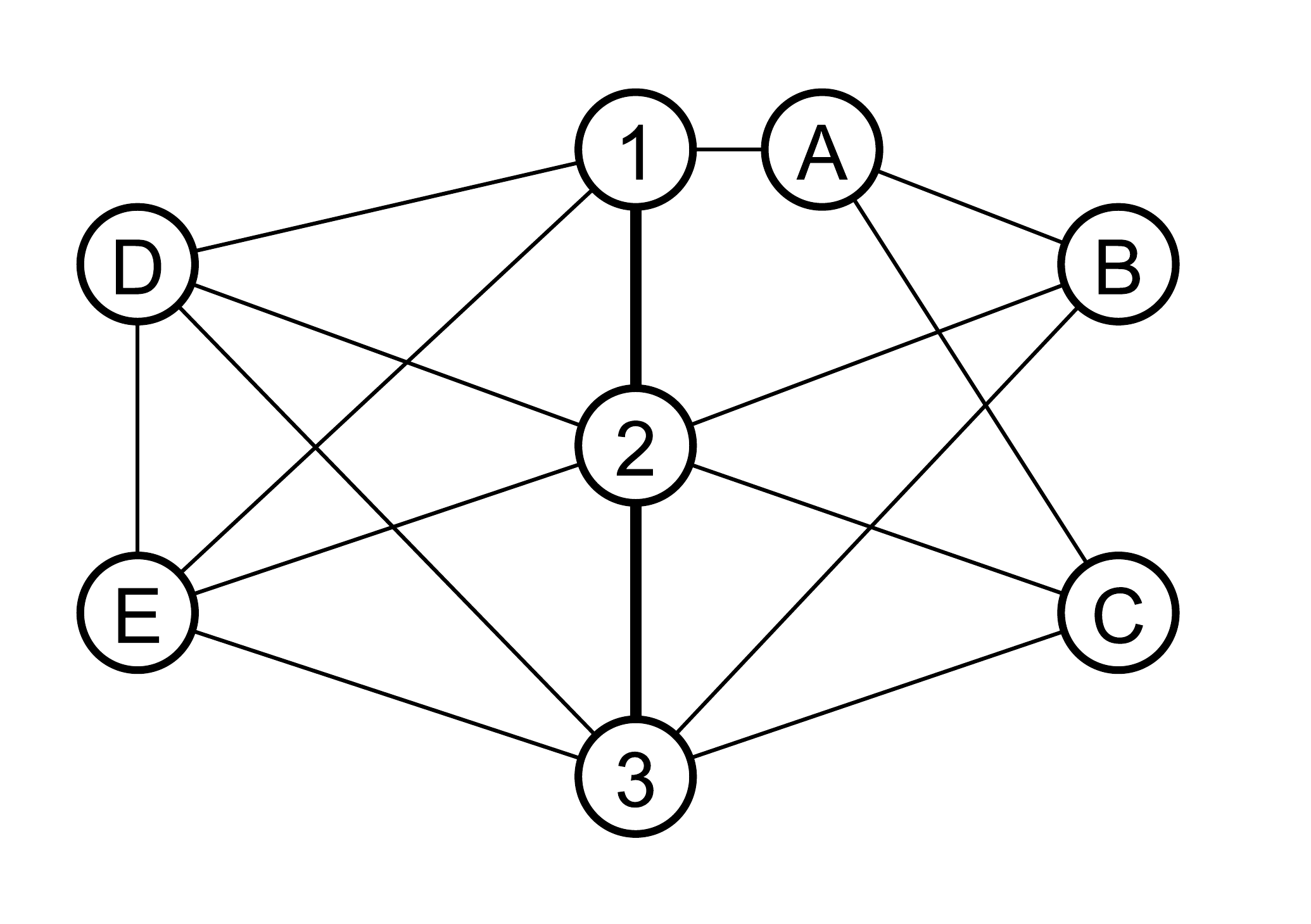}
\caption{\DD}
\label{fig:D3}
\end{figure}

\begin{figure}[htbp]\centering
\includegraphics[height=4.5cm]{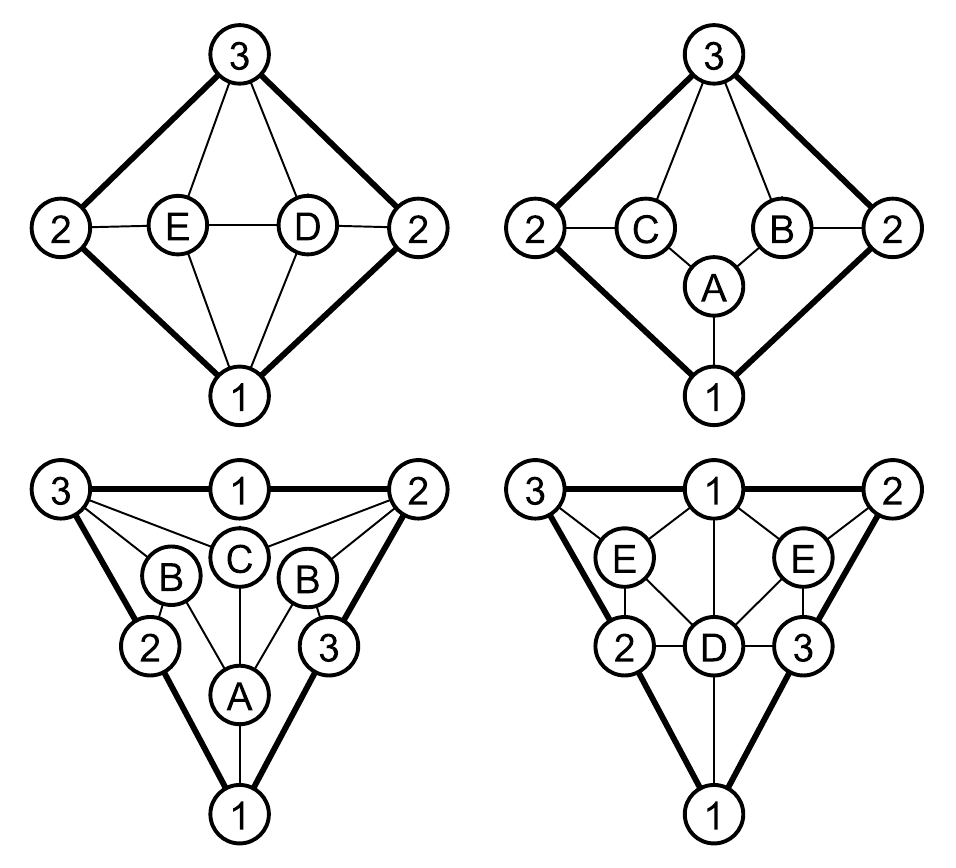}
\caption{Building blocks for \DD emulator.}
\label{fig:D3_cells}
\end{figure}

\begin{figure}[htbp]
\centering
\includegraphics[height=8cm]{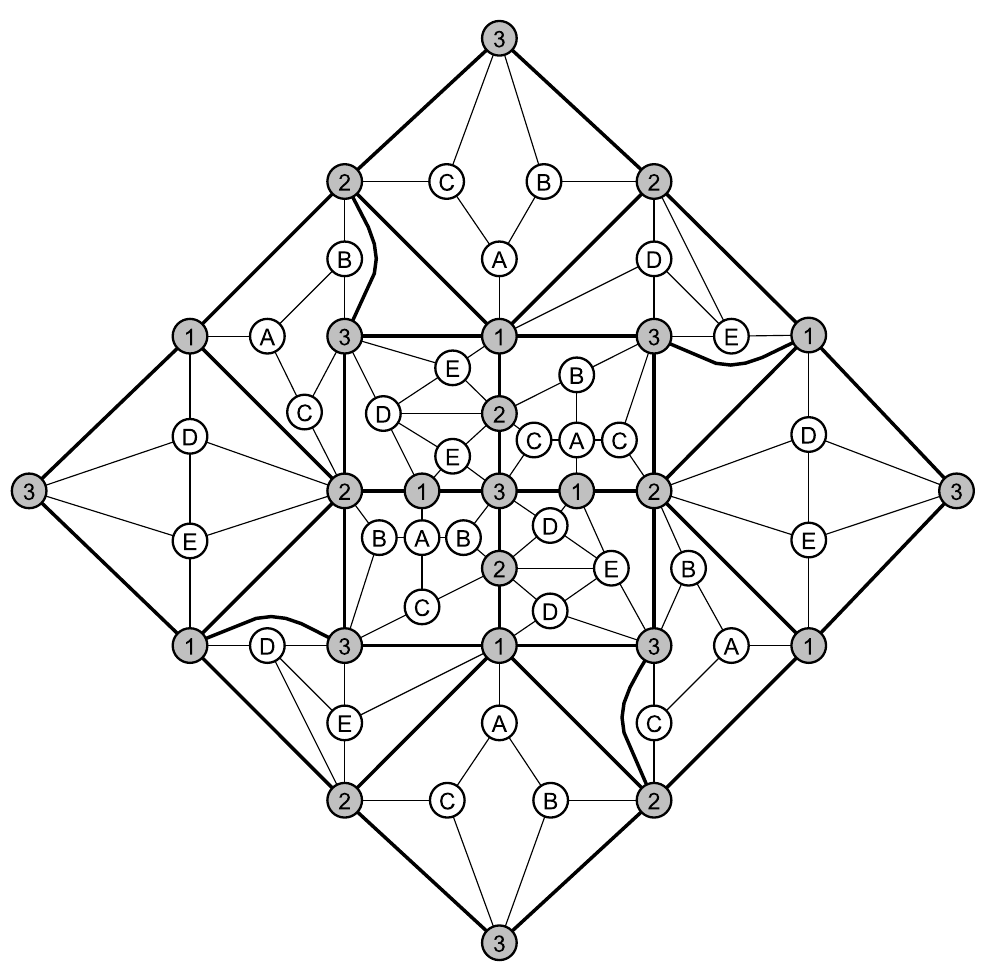}
\caption{The construction built with one half of the emulator for \KK and 8 small cells for the outer vertices to have the maximal number of different
neighbors.}
\label{fig:D3_help}
\end{figure}

\begin{figure}[htbp]
\centering\smallskip
\includegraphics[height=2.8cm]{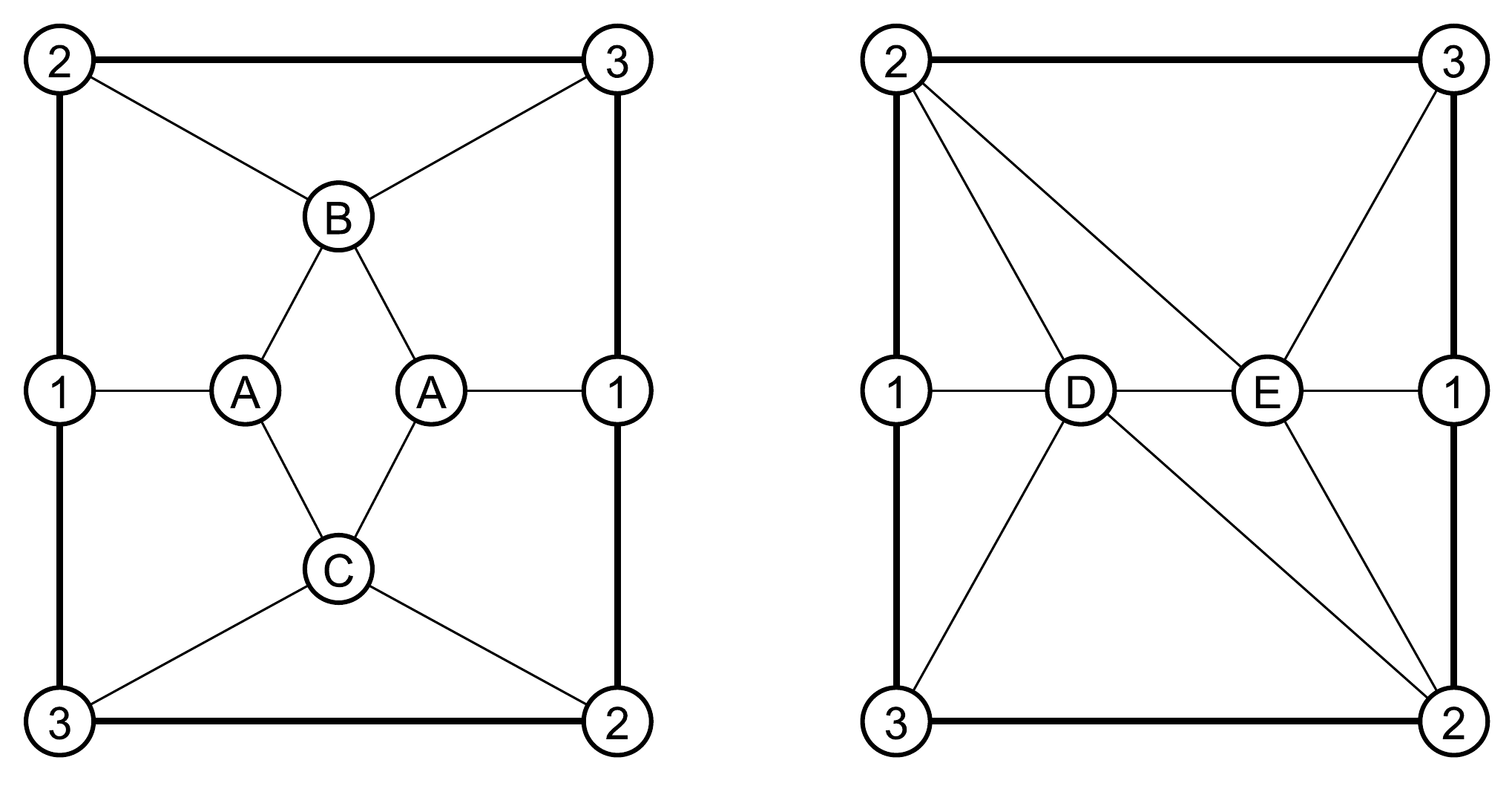}
\caption{The hexagonal cell for connecting two identical components from  Figure~\ref{fig:D3_help} into an
\DD emulator.}
\label{fig:D3_hex}
\end{figure}

\begin{figure}[htbp]
\includegraphics[width=12cm]{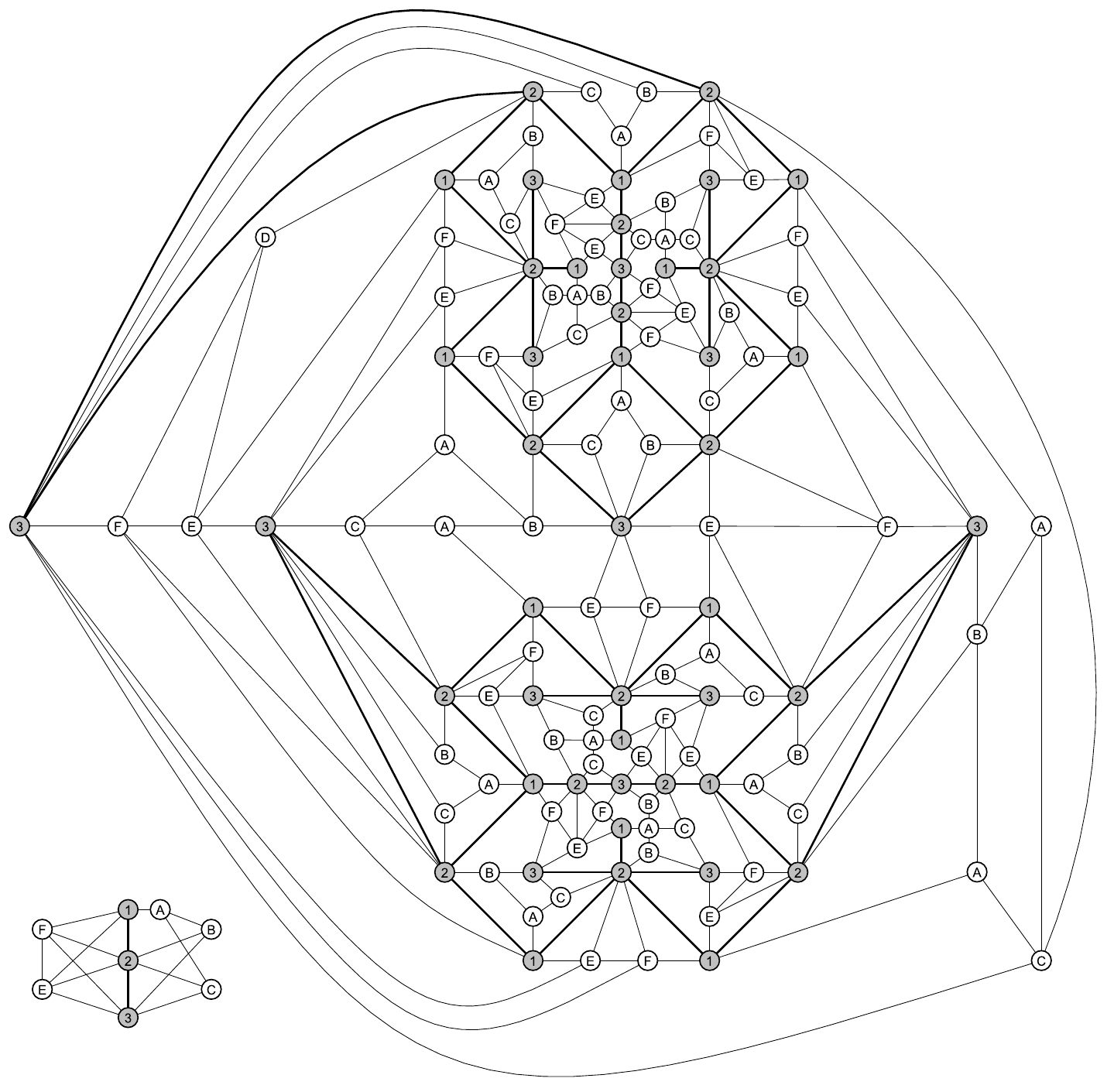}
\caption{The finite planar emulator for \DD.}
\label{fig:D3_final}
\end{figure}

While building \KK we used two triangular gadgets \textit{AB cell} and
\textit{CD cell}.  If we upgrade one of them in order to get a cell suitable
for the \textit{ABC component} of \DD (a cell with \textit{central vertices} as
outer vertices, satisfying all the inner and the maximum of outer vertices)
and try to establish the emulator in the exactly same way as for \KK (using
an octahedron), we arrive at a single, but fatal obstruction - the vertex 1
cannot be \textit{half-satisfied} by the \textit{ABC cell}, simply because it only has one neighbor among A, B and C.  Therefore no vertex 1 on
the corner of an octahedron can meet its requirements.  Nevertheless, we can
take a \KK emulator as a core and "fix" the properties of such vertices.

Let's have two building blocks as in the \KK case and define two other
supporting cells, \textit{ABC-small-cell} and \textit{DE-small-cell} (see
Figure~\ref{fig:D3_cells}).  These will help us overcome the above
mentioned drawback.  We take one half of the emulator for \KK and upgrade the
two \textit{AB cells} to \textit{ABC cells}.  Now we surround the graph with
four \textit{ABC-small-cells} and four \textit{DE-small-cells}, such that
the \textit{central vertices} have all desired neighbors among $\{1,2,3\}$
and the outer vertices of the new expanded graph have better properties
concerning the number and kind of neighbors (see Figure~\ref{fig:D3_help}). 
We can observe that all vertices labeled with 3 are \textit{ABC-satisfied}
or \textit{DE-satisfied}, but do not have one of the \textit{central
vertices} as a neighbor, vertices labeled with 1 are \textit{ABC-satisfied}
and \textit{DE-half-satisfied} or \textit{DE-satisfied} and vertices with
label 2 are \textit{DE-satisfied} and \textit{ABC-half-satisfied} or
\textit{ABC-satisfied}.  Additionally, if there is an edge between marginal
vertex 1 and 2, then they miss some neighbors from the component of the same
kind.  This fact enables us to copy the whole graph in
Figure~\ref{fig:D3_help} and join it in a smart way with the original graph
to obtain an emulator.  If we identify pairs of vertices with label 3 from
those two copies (such that vertices that have only vertices 1 as a neighbor
will be identified with the vertex which have the opposite problem), we get
4 empty hexagons, whose borders are made of two vertices with label 3, which
are \textit{completely satisfied}, two vertices with label 1, which are
\textit{ABC-satisfied} and \textit{DE-half-satisfied} and two vertices with
label 2, which are \textit{ABC-satisfied} (it holds for two of those
hexagons, for the other two the \textit{ABC and DE-satisfactory} properties
are switched).  Thus these hexagons can be easily filled with the simple
pattern for the border vertices to meet the required conditions (see
Figure~\ref{fig:D3_hex}).

The final emulator for \DD is presented in Figure~\ref{fig:D3_final} and the
above described approach is clearly visible (two identical components derived from an emulator for
\KK connected together). 

\begin{figure}[tb]
\includegraphics[width=12cm]{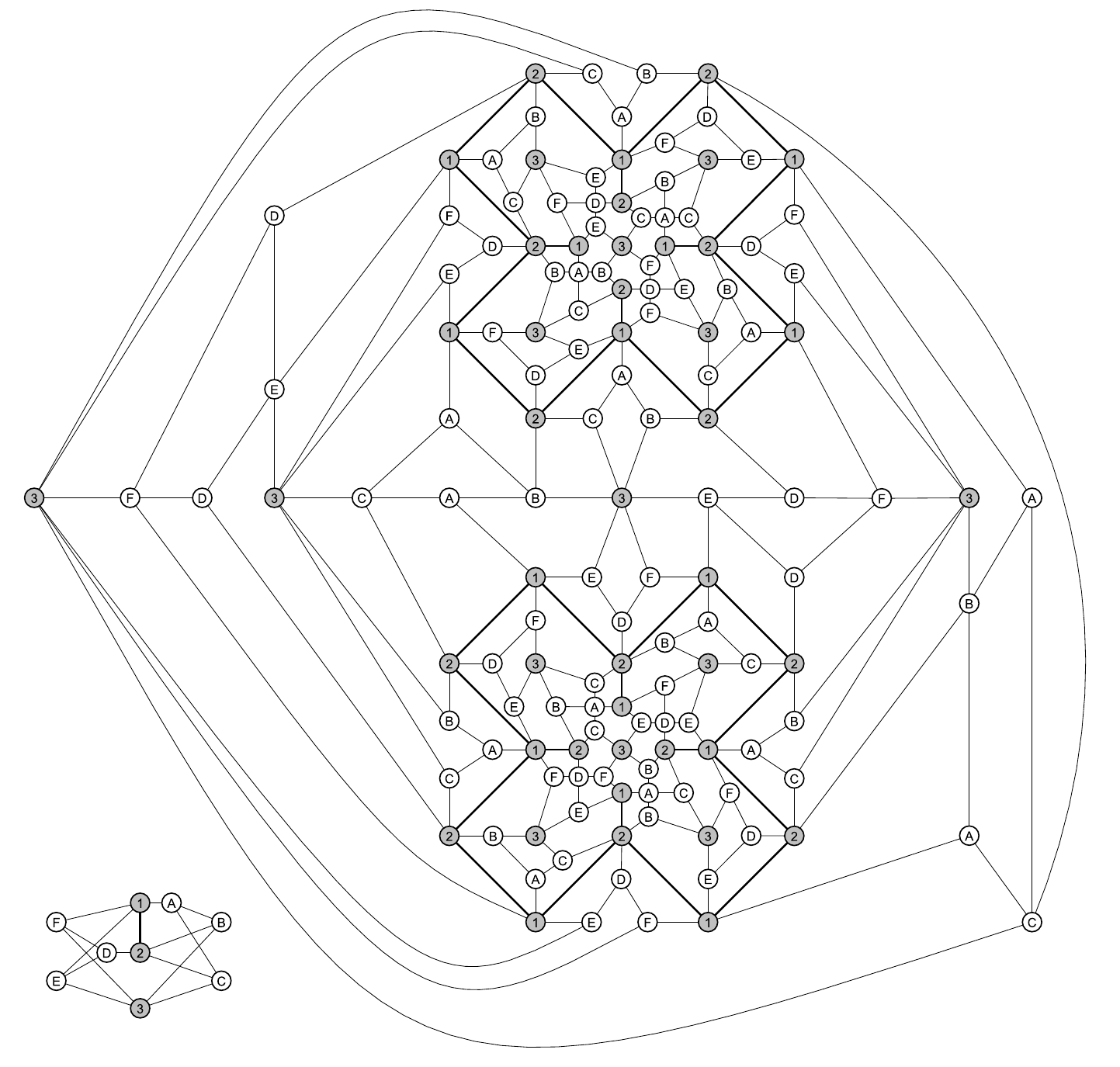}
\caption{The finite planar emulator for \FF}
\label{fig:F1_final}
\end{figure}

\subsubsection{Planar emulator for \FF.}

The construction of an emulator for \FF follows the same pattern as building
an emulator for \DD.  In fact, the emulator for \FF was found first by the
above mentioned construction and the emulator for \DD resulted from a
simplification of an emulator for \FF, (\DD results from \FF by taking a
{\ensuremath{Y \Delta}} transformation, which is trivial to perform in the
emulator).  Therefore we present only the final emulator picture (see
Figure~\ref{fig:F1_final}).

\subsubsection{Planar emulator for \EEE.}

\begin{figure}[bt]
\centering
\includegraphics[height=5cm]{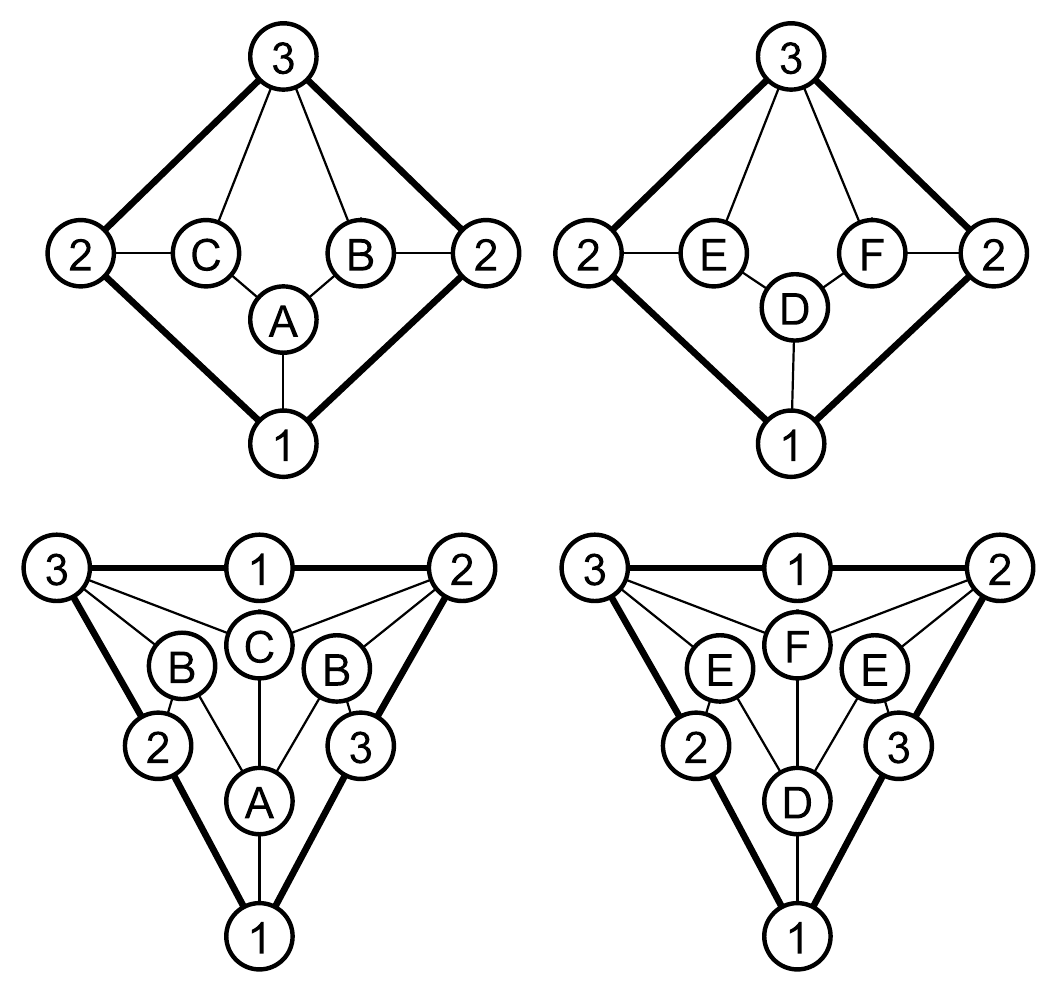}
\caption{Building cells for \EEE emulator}
\label{fig:E5_cells}
\end{figure}

\begin{figure}[tbp]
\centering
\includegraphics[height=6.5cm]{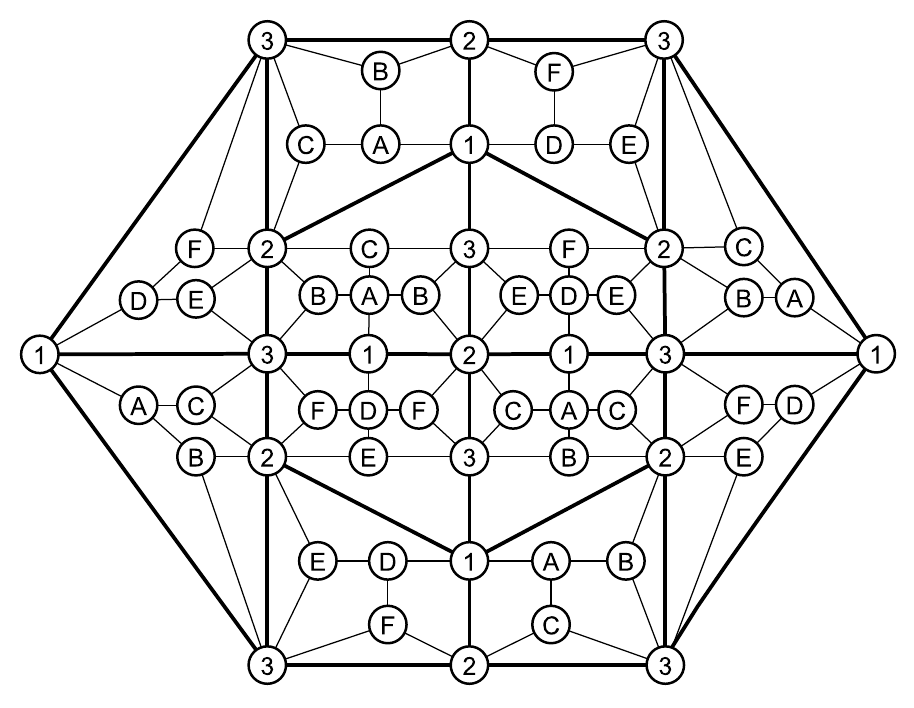}
\caption{The construction for \EEE built upon a ``half'' of a 
	\KK emulator and 8 small cells for the outer vertices to have the best possible properties}
\label{fig:E5_help}
\end{figure}

In order to obtain an emulator for \EEE, we again take one half of the emulator for \KK
(as in the \DD case) and replace \textit{AB/CD cells} by
\textit{ABC/DEF cells}.  Let's call this construction a \textit{core}.  As
in \DD case, we consider an additional edge 13, which is not present in
\EEE but makes the pictures easier to understand.  Similarly, we use some
smaller additional cells to improve the properties of the outer vertices of
the \textit{core} (see Figure~\ref{fig:E5_cells}).  Since \EEE is slightly
different from \FF (both come from \KK, but two {\ensuremath{\Delta Y}}
transformations took place in different triangle faces, so there exists a
vertex (labeled 2) in \EEE which is adjacent to two new vertices of
degree 3, but is not present in \FF), the use of the supporting
small cells is quite different as well.  We surround the \textit{core} as showed in Figure~\ref{fig:E5_help}.  In this way we arrive to better
properties of the outer vertices of the new graph.  Outer vertices labeled 1 are \textit{completely satisfied}, vertices 2 are
\textit{half-ABC-satisfied} and \textit{half-DEF-satisfied} and vertices 3
are \textit{ABC-satisfied} and \textit{DEF-half-satisfied} or vice versa.

\begin{figure}[tbp]
\includegraphics[width=12cm]{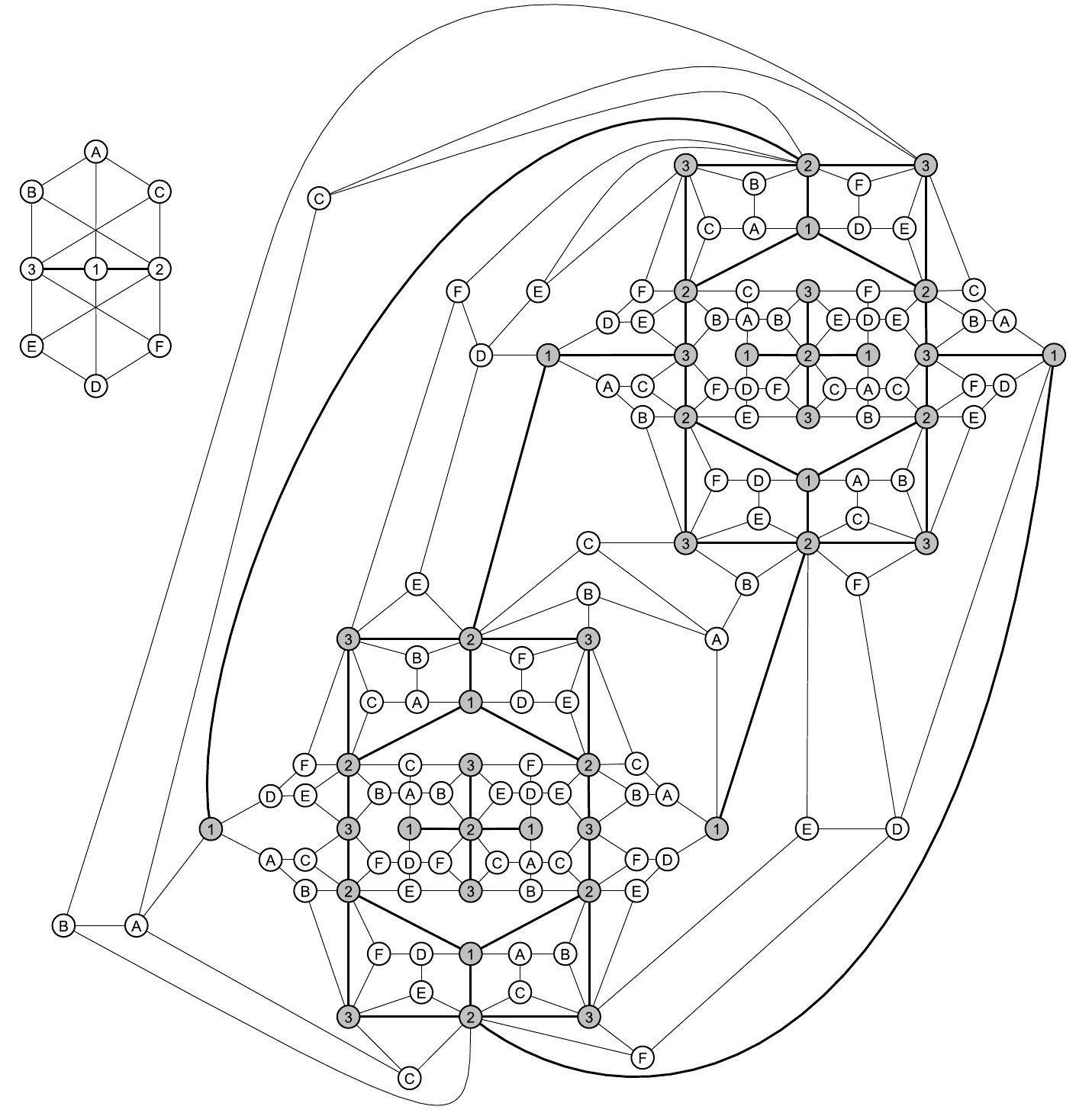}
\caption{The finite planar emulator for \EEE}
\label{fig:E5_final}
\end{figure}

Now we use a similar tool as in the previous cases of \DD and \FF -- we
duplicate the graph and connect the two copies in a clever way so that
vertices 1 get the desired neighbor 2 as well and four new hexagons are
created.  The vertices surrounding each hexagonal face are missing some
neighbors of the same component.  We fill each face with a \textit{ABC or
DEF cell} to satisfy all the remaining vertices.  Now we have a complete
emulator for \EEE (see Figure~\ref{fig:E5_final}).


\end{accumulate}

\section{Structural Search: How far can we go?}
\label{sec:algo}
\label{sec:forbminors}

Until now, we have presented several newly discovered planar
emulators of nonprojective graphs.
Unfortunately, despite the systematic construction methods introduced in
Section~\ref{sec:construct}, we have got nowhere closer to
a real understanding of the class of planar-emulable graphs.
It is almost the other way round---the new planar emulators evince 
more and more clearly how complicated the problem is.
Hence, we also need to consider a different approach.

The structural search method, on which we briefly report in this section,
is directly inspired by previous \cite{cit:counterex};
we refer to~\cite{cit:martinbc,cit:florida}
for closer details which cannot fit into this paper.

The general idea can be outlined as follows:
If $H$ is a mysterious nonprojective planar-emulable graph,
then $H$ must contain one of the projective forbidden minors, say $F$,
while $F$ cannot be among those forbidden minors not having planar emulators
(Theorems~\ref{thm:2kgraphs},~\ref{thm:K35noemul}).
Now there are basically three mutually exclusive possibilities:
\begin{enumerate}[i.]\vspace{-3pt}
\item
$H$ is a planar expansion of a smaller graph.
A graph $H$ is a \emph{planar expansion} of $G$ if it can be obtained by
repeatedly substituting a vertex of degree $\leq3$ in $G$ 
by a planar subgraph with the attachement vertices on the outer face. 
\item
$H$ contains a nonflat $3$-separation.
A separation in a graph is called {\em flat} if one of the sides
has a plane drawing with all the boundary vertices on the outer face.
\item
$H$ is {\em internally $4$-connected}, i.e.,
it is $3$-connected and each $3$-separation in $H$ has one side inducing the
subgraph $K_{1,3}$
(informally, $H$ is $4$-connected up to possible degree-$3$ vertices
with stable neighborhood).
\end{enumerate}

We denote by
$\langle K_7-C_4 \rangle=\{\KK,\DD, \EEE, \FF\}$
the {\em family of $K_7-C_4$}.
The underlying idea is that all the graphs in a family
are $Y\!\Delta$-transformable to the family's base graph.
Particulary the family of $K_7-C_4$ comprises all the projective
forbidden minors in question which are not internally $4$-connected.
See in \figurename~\ref{fig:32nonproj}.

In the case (i.) above, we simply pay attention to the smaller graph $G$.
In the case (ii.), one can argue that either the projective forbidden minor
$F$ (in~$H$) itself contains a nonflat $3$-separation (so
$F\in\langle K_7-C_4 \rangle$),
or $F$ is internally $4$-connected and $H$ then is not planar-emulable 
(a contradiction).
The former is left for further investigation.
Finally, in the case (iii.) we may apply a so-called {\em splitter theorem}
for internally $4$-connected graphs \cite{cit:splitter},
provided that $F$ is also internally $4$-connected.
This leads to a straightforward computerized search which has a
high chance to finish in finitely many steps,
producing all such desired internally $4$-connected graphs~$H$.

Actually, when the aforementioned procedure was applied to the planar cover
case in \cite{cit:counterex}, the search was so efficient that the outcome
could have been described by hand; giving all $16$ specific graphs that
potentially might be counterexamples to Conjectures~\ref{conj:negami}.
In our emulator case, we get the following:

\begin{theorem}[\cite{cit:florida}]
\label{thm:our_result}
Let $H$ be a nonprojective planar-emulable graph. 
Then, $H$ is a planar expansion of one of specific $175$ 
internally 4-connected graphs, or $H$ contains a minor isomorphic
to a member of $\{\EE,K_{4,5}-4K_2\}\cup \langle K_7-C_4 \rangle$.
\end{theorem}

Up to this point, we have not been successful in finishing the computations 
for the graphs $F= K_{4,5} - 4K_2$ and $\mathcal{E}_2$,
due to the high complexity of the generated extensions.
Yet, we strongly believe that it is possible to obtain finite results
also for those cases, perhaps with the help of an improved generating procedure.
On the other hand, the cases starting with $F\in\langle K_7-C_4\rangle$
will need an alternative procedure, e.g., using so-called ``separation bridging''.
This is subject to future investigations.


\section{Conclusion and Further Questions}
\label{sec:conclus}

While our paper presents new and surprising findings about
planar-emulable graphs, the truth is that these findings are often negative
in the sense that they bring {\em more intriguing questions than answers}.
Of course, the fundamental open question in the area is to find a
characterization of the class of planar-emulable graphs in terms of some
other natural (and preferably topological) graph property.
Even coming up with a plausible conjecture (cf.~Conjecture~\ref{conj:negami})
would be of high interest, but, with
our current knowledge, already this seems to be out of reach yet.

Instead, we suggest to consider the following specific (sub)problems:
\begin{itemize}\parskip2pt
\item
Is there a planar emulator of the graph $K_{4,4}-e$?
We think the answer is {\em no}, but are currently unable to find a proof,
e.g.\ extending the arguments of~\cite{cit:k44-e}.
\item
The emulators shown in Section~\ref{sec:construct} suggest that
we can, in some mysterious way, reflect $\Delta Y$-transformations in
emulator constructions (i.e., the converse direction of
Proposition~\ref{prop:clYDelta}).
Such a claim cannot be true in general since, e.g.,
a $Y\!\Delta$-transformation of the graph $\ca D_4$
(\figurename~\ref{fig:32nonproj}) leads to a strict subgraph of $\ca B_3$,
which therefore has a two-fold planar cover while $\ca D_4$ is not
planar-emulable by Theorem~\ref{thm:2kgraphs}.
But where is the precise breaking point?
\item
The two smallest projective forbidden minors are on $7$
vertices, $K_7-C_4$ (missing four edges of a cycle) 
and $K_{1,2,2,2}$ (missing three edges of a matching).
Both of them, however, have planar emulators while their common supergraph
$K_7$ does not.
What is a minimal subgraph of $K_7$ not having a planar emulator?
Can we, at least, find a short argument that the graph $K_7-e$ has no planar
emulator?
\item
Finally, Conjecture~\ref{conj:negami} can be reformulated in a way that a
graph has a planar cover iff it has a two-fold planar cover.
The results of \cite{cit:counterex} moreover imply that the minimal
required fold number for planar-covers is bounded by a constant.
Although, in the emulator case, the numbers of representatives for each
vertex of the emulated graph differ, there is still a possibility of a fixed
upper bound on them:
Is there a constant $K$ such that every planar-emulable graph $H$ has a
planar emulator with projection $\psi$ such that $|\psi^{-1}(v)|\leq K$ for
all $v\in V(H)$?
A computerized search as in Section~\ref{sec:forbminors} would be of great
help in this task.
\end{itemize}


\begin{onlynoaccum}
\newpage
\section*{\LARGE Appendix}
\def\thesection{A}

\dotheappendixmagic
\end{onlynoaccum}

\end{document}